\documentclass[aps,prd,reprint,groupedaddress]{revtex4-2}

\usepackage[utf8]{inputenc}
\usepackage{physics}
\usepackage{amsmath,amsthm,amssymb}
\usepackage{mathrsfs}
\usepackage{graphicx}
\usepackage{dcolumn}
\usepackage{bm}
\usepackage{siunitx}
\usepackage{caption}
\usepackage{subcaption}
\usepackage[hidelinks]{hyperref}

\newtheorem{theorem}{Theorem}

\newcommand{\ee}{\mathrm{e}}

\begin{document}
	\title{Asymptotic Throat as a Minimal Geometric Alternative to the Singularity in Schwarzschild Black Hole}
	
	\author{Yi-Bo Liang}
	\email[]{liangyibo@stu.xjtu.edu.cn}
	\author{Hong-Rong Li}
	\email[Corresponding author: ]{hrli@xjtu.edu.cn}
	\affiliation{School of Physics, Xi’an Jiaotong University, Xi’an 710049, China}
	
	\date{\today}
	\begin{abstract}
		A basic question in the study of regular black holes is whether the Schwarzschild singularity can be replaced by a geodesically complete interior while retaining a single horizon.
		An asymptotic throat offers a minimal geometric realization of this possibility.
		The areal radius approaches a finite, nonzero limit at a future internal infinity, without requiring a change in spatial topology or a bounce.
		We extend the geometric analysis beyond monotone contraction to interiors in which the areal radius oscillates about its limit.
		For both null-type and spacelike-type inner ends, we derive a necessary and sufficient integral criterion for causal geodesic completeness and criteria for bounded curvature in parallel-propagated frames.
		Explicit examples of both types have complete inner ends and satisfy all four standard energy conditions at sufficiently large exterior radii.
		Applying these criteria to the oscillatory solution of the B\"ohmer--Vandersloot model of effective loop quantum gravity, we find a causally complete spacelike-type inner end.
		Its scalar curvature invariants converge to those of the limiting Nariai-type solution.
		Along each radial null geodesic, however, a parallel-propagated curvature component is unbounded along a sequence of turning points as the affine parameter tends to infinity.
	\end{abstract}
	
	\maketitle
	\section{Introduction}
	The Schwarzschild solution \cite{Schwarzschild1916} describes the vacuum exterior of a spherically symmetric body and provides a basic model of a nonrotating, uncharged, static black hole in general relativity.
	Its simple and exact geometry has made it a standard setting for understanding event horizons and causal structure, and for exploring black hole thermodynamics and quantum effects in curved spacetime.
	At the same time, it exposes a fundamental limitation of the classical description: infalling observers reach a curvature singularity \cite{penrose1965gravitational,hawking1967occurrence,hawking1970singularities} in finite proper time.
	If this singularity signals a breakdown of classical general relativity, a quantum theory of gravity is expected to resolve it.
	
	Given the difficulty of developing a complete theory of quantum gravity, effective geometric models are often used to study the geometric conditions for singularity resolution without specifying a microscopic theory.
	Two common approaches to resolving the Schwarzschild singularity are regular-center models \cite{bardeen1968non,ayon1998regular,ayon2000bardeen,bronnikov2001regular,burinskii2002new,fan2016construction,bronnikov2023regular}, which involve topology change and contain a Cauchy horizon, and black-bounce models \cite{simpson2019black,lobo2021novel,bronnikov2022black,bronnikov2022field,canate2022black,ashtekar2023regular}, in which evolution continues into another universe through a bounce.
	Other proposals include finite-radius cylindrical ends in scalar-field models \cite{bronnikov2006regular}, asymptotic hidden wormholes identified in a classification of geodesically complete black holes \cite{carballo2020geodesically}, and a Nariai-type interior arising in an effective loop quantum gravity model \cite{boehmer2007loop}.
	
	Asymptotic throats provide a minimal and general approach to resolving the Schwarzschild singularity.
	Studies of transitions between dynamical black holes and traversable wormholes  \cite{hayward1999dynamic, shinkai2002fate, hayward2009wormhole, hayward2006formation, simpson2019vaidya, lobo2020dynamic, yang2021trapping} motivated the introduction of the asymptotic throat as a null-type ideal boundary \cite{liang2025globally}.
	Ref.~\cite{LiangLi2026AT} later developed this idea into an explicit geometric construction.
	The starting point is the assumption that the areal radius approaches a positive limit rather than zero.
	Inside the Schwarzschild horizon, the areal radius serves as a timelike coordinate.
	This motivates replacing the singularity with a future asymptotic end at a finite, nonzero limiting areal radius.
	
	The next step is to determine the properties of asymptotic throat geometries.
	We consider two types of asymptotic throats: null-type and spacelike-type.
	A positive limiting radius alone guarantees neither geodesic completeness nor curvature regularity.
	Bounded scalar curvature invariants need not imply bounded curvature in freely falling frames, as earlier studies of curvature singularities have shown \cite{Zaslavskii2007,Bronnikov2008}.
	A further question is whether these interior conditions are compatible with a regular horizon and an asymptotically flat exterior satisfying the standard energy conditions sufficiently far from the black hole. 
	Answering these questions requires both general criteria and explicit global constructions.
	Finally, we ask whether asymptotic throat geometries can arise in effective models of quantum gravity.
	
	In this work, we study spherically symmetric geometries with a static exterior, a single simple horizon, and a positive limiting areal radius in the interior.
	The approach to this limit need not be monotonic.
	We derive a necessary and sufficient integral criterion for causal geodesic completeness toward the inner end, applicable to both null-type and spacelike-type throats.
	A separate integral determines the causal character of that end.
	Under the stated boundedness and finite-limit assumptions on the independent curvature components, the limiting curvature of a complete null-type throat is fixed by its limiting areal radius.
	A complete spacelike-type throat allows an additional nonnegative curvature parameter.
	We also derive conditions for bounded curvature in parallel-propagated frames along timelike and null geodesics.
	The radial null energy condition must be violated somewhere in the interior.
	Then, we give explicit real-analytic constructions with either type of complete throat, a regular horizon, and an asymptotically flat exterior.
	In these geometries, the null, weak, dominant, and strong energy conditions all hold sufficiently far from the black hole.
	We also construct smooth geometries that become exactly Schwarzschild beyond a finite radius.
	Finally, we apply these criteria to the Böhmer--Vandersloot model and show that its oscillatory interior approaches a complete spacelike-type asymptotic throat.
	
	\section{Geometric framework}
	\label{sec2}
	
	We use the signature $(-,+,+,+)$ and units $G=c=1$.
	We define the effective stress tensor by $T^{\mathrm{eff}}_{ab}=G_{ab}/(8\pi)$.
	All energy densities and pressures below refer to this effective stress tensor.

	Ref.~\cite{LiangLi2026AT} assumes a global metric of the form
	\begin{equation}
		ds^2=-A(\zeta)(\dd T^2-\dd X^2)+r^2(\zeta)\dd\Omega^2,\label{globalmetric}
	\end{equation}
	where $\dd \Omega^2$ is the metric of the unit two-sphere and
	\begin{gather}
		\zeta=X^2-T^2 \in (\zeta_i,\infty),\quad \zeta_i\in[-\infty,0),\\
		A(\zeta)>0,\quad
		\frac{\mathrm{d}r}{\dd\zeta}>0,\quad
		r:(\zeta_i,\infty) \longrightarrow (L, \infty).
	\end{gather}
	In that framework, $A$ and $r$ are smooth at every finite point of their domain.
	The condition $\dd r/\dd\zeta>0$ ensures that the spheres in the future interior $T>|X|$ are trapped, while those in the exterior $\zeta>0$ are untrapped.
	It also guarantees that the inverse function $\zeta(r)$ is smooth and satisfies $\zeta'(r)>0$.
	A prime denotes differentiation with respect to $r$ unless otherwise stated.
	The inner end is of null type when $\zeta_i=-\infty$ and of spacelike type when $-\infty<\zeta_i<0$.
	The metric admits the Killing field
	\begin{equation}
		\xi = \kappa (X\partial_T+T\partial_X),\quad \kappa>0,
	\end{equation}
	whose normalization is fixed by requiring $\xi^a\xi_a\to-1$ at the asymptotically flat end.
	Its squared norm is
	\begin{equation}
		\xi^a\xi_a=-\kappa^2 A(\zeta)\zeta.
	\end{equation}
	The Killing horizon lies at $\zeta=0$ and is regular.
	Its future branch adjoining the right exterior $X>|T|$ has surface gravity $\kappa$.
	
	In the exterior, introduce
	\begin{equation}
		T=\sqrt{\zeta}\sinh(\kappa t),\qquad
		X=\sqrt{\zeta}\cosh(\kappa t),
	\end{equation}
	so that $\xi=\partial_t$.
	The metric becomes
	\begin{equation}
		ds^2=-\kappa^2 A \zeta \dd t^2+\frac{A \zeta^{\prime 2}}{4\zeta}\dd r^2+r^2\dd\Omega^2.\label{localmetric}
	\end{equation}
	Define
	\begin{equation}
		D(r) = \frac{4\zeta}{A\zeta^{\prime 2}},\quad \ee^{\psi(r)} = \frac{\kappa A \zeta^\prime}{2}.\label{Dpsi}
	\end{equation}
	Then
	\begin{equation}
		ds^2=-D(r)\ee^{2\psi(r)}\dd t^2+D(r)^{-1}\dd r^2+r^2\dd\Omega^2.\label{localmetric2}
	\end{equation}
	The same functions describe the interior, where $D<0$.
	Let $r_h=r(0)$.
	Since $A$ and $\zeta^\prime$ are finite and positive at the horizon,
	\begin{equation}
		D(r_h)=0,\quad D^\prime(r_h)>0,\quad \frac{1}{2}D^\prime(r_h) \ee^{\psi(r_h)}=\kappa.\label{DDprimepsi}
	\end{equation}
	Thus $D$ has a simple zero at the horizon.
	The strict monotonicity of $r(\zeta)$ and its positive inner limit imply $0<L<r_h$.
	Since $D$ has a single simple zero at $r_h$, with $D<0$ for $L<r<r_h$ and $D>0$ for $r>r_h$, it can be written uniquely as
	\begin{equation}
		D(r)=\left(1-\frac{r_h}{r}\right)\ee^{U(r)},
	\end{equation}
	where $U(r)$ is smooth on $(L,\infty)$.
	The horizon relation then becomes
	\begin{equation}
		U(r_h)+\psi(r_h)=\ln(2\kappa r_h).
	\end{equation}
	
	To allow for possible corrections to the classical geometry in the interior, we relax the global monotonicity assumption on $r(\zeta)$ while retaining $\dd r/\dd\zeta>0$ in the exterior and near the horizon.
	We retain the global metric and its endpoint limits, with $A$ and $r$ smooth and positive throughout their domain.
	The areal radius remains bounded away from zero, although its limiting value $L$ need not be a lower bound.
	The description in terms of $D$, $\psi$, and $U$, together with the horizon relations above, remains valid in the exterior and near the horizon.
	In the interior, the areal radius is used as a coordinate only on intervals where $\dd r/\dd\zeta\ne0$.
	On each such interval, $\ee^\psi=\kappa A|\zeta^\prime|/2$.
	In the future interior $T>|X|$, the symmetry spheres are trapped where $\dd r/\dd\zeta>0$ and anti-trapped where $\dd r/\dd\zeta<0$.
	Zeros of $\dd r/\dd\zeta$ in the interior do not introduce additional Killing horizons.
	The corresponding constant-$\zeta$ hypersurfaces remain spacelike.
	Both statements follow directly from Eq.~\eqref{globalmetric}.
	
	We therefore introduce coordinates that remain valid at stationary points of the areal radius. In the future interior, set
	\begin{equation}
		T=\sqrt{-\zeta}\cosh(\kappa t),\quad
		X=\sqrt{-\zeta}\sinh(\kappa t).
	\end{equation}
	Define an interior time coordinate $\tau$ and a scale factor $a$ by
	\begin{equation}
		\frac{\dd\tau}{\dd\zeta}
		=-\frac{\sqrt{A(\zeta)}}{2\sqrt{-\zeta}},
		\qquad
		a(\tau)=\kappa\sqrt{-A(\zeta(\tau))\zeta(\tau)}.
		\label{taur}
	\end{equation}
	Then $\tau$ increases toward the inner end, and the metric becomes
	\begin{equation}
		\dd s^2=-\dd\tau^2+a(\tau)^2\dd t^2
		+r(\tau)^2\dd\Omega^2.
		\label{taut}
	\end{equation}
	A dot denotes differentiation with respect to $\tau$. 
	In particular,
	\begin{equation}
		\dot r=-2\sqrt{\frac{-\zeta}{A}}\frac{\dd r}{\dd\zeta},
	\end{equation}
	which may vanish or change sign.
	Choose a reference point $\zeta_0\in(\zeta_i,0)$, with $\tau_0=\tau(\zeta_0)$, and denote the possibly infinite limiting value of $\tau$ at the inner end by $\tau_i$.
	
	\section{Asymptotic Throats in Black Hole Interiors}
	\label{sec3}
	
	\subsection{Causal character and geodesic completeness}
	
	For the interior metric \eqref{taut}, define
	\begin{equation}
		I_\star\equiv\int_{\tau_0}^{\tau_i}
		\frac{\dd\tau}{a(\tau)},
		\label{Istar}
	\end{equation}
	and
	\begin{equation}
		I_{\mathrm{geo}}\equiv\int_{\tau_0}^{\tau_i}
		\frac{\dd\tau}{\sqrt{1+a(\tau)^{-2}}}.
		\label{Igeo}
	\end{equation}
	As proved in Theorem~\ref{thm1}
	in Appendix~\ref{app1}, the inner end is of null type if and only if $I_\star=\infty$, and of spacelike type if and only if $I_\star<\infty$.
	Causal geodesic completeness toward this end is equivalent to $I_{\mathrm{geo}}=\infty$.
	For a spacelike-type end, this condition reduces to $\tau_i=\infty$.
	
	\subsection{Curvature invariants}
	
	For the interior metric \eqref{taut}, define
	\begin{equation}
		\begin{aligned}
			k_1&=-\frac{\ddot a}{a}, &
			k_2&=-\frac{\dot a\dot r}{ar},\\
			k_3&=-\frac{\ddot r}{r}, &
			k_4&=\frac{1+\dot r^2}{r^2}.
		\end{aligned}
		\label{eq:interior-curvature}
	\end{equation}
	These functions determine the independent components of the Riemann tensor in the orthonormal frame adapted to $\partial_\tau$.
	With our curvature convention, the Ricci scalar $R$, the Ricci contraction $\mathcal{S}=R_{ab}R^{ab}$, and the Kretschmann scalar $\mathcal{K}=R_{abcd}R^{abcd}$ are
	\begin{equation}
		\begin{gathered}
			R=-2k_1-4k_2-4k_3+2k_4,\\
			\mathcal{S}=(k_1+2k_2)^2+(k_1+2k_3)^2
			+2(k_4-k_2-k_3)^2,\\
			\mathcal{K}
			=4(k_1^2+2k_2^2+2k_3^2+k_4^2).
		\end{gathered}
		\label{eq:curvature-scalars}
	\end{equation}
	Thus $\mathcal{K}$ is bounded if and only if all four $k_i$ are bounded.
	In that case, $R$ and $\mathcal{S}$ are also bounded.
	
	Theorem~\ref{thm2} in Appendix~\ref{app1} gives the curvature limits at a causally complete inner end with $r\to L>0$.
	If $\mathcal{K}$ is bounded, then $\dot r\to0$, $k_2\to0$, and $k_4\to L^{-2}$ for both types of end.
	If $k_1$ and $k_3$ also have limits, then $k_3\to0$ and $k_1\to-\beta$, where $\beta\geq0$. Hence
	\begin{equation}
		\begin{aligned}
			R&\to 2(\beta+L^{-2}),\\
			\mathcal{S}&\to 2(\beta^2+L^{-4}),\\
			\mathcal{K}&\to 4(\beta^2+L^{-4}).
		\end{aligned}
		\label{sRSK}
	\end{equation}
	For a null-type end, $\beta=0$, so $L$ fixes all three limits.
	A spacelike-type end allows $\beta>0$ as well.
	
	The null-type example in Ref.~\cite{LiangLi2026AT} has $\beta=0$,
	and the spacelike-type example has $\beta=(2M-L)/L^3$.
	For comparison, Nariai spacetime is the direct product of two-dimensional de Sitter spacetime and a two-sphere, both with curvature $\Lambda>0$ \cite{Nariai1999}.
	Its invariants are
	\begin{equation}
		R=4\Lambda,\qquad
		\mathcal{S}=4\Lambda^2,\qquad
		\mathcal{K}=8\Lambda^2.
	\end{equation}
	All three values agree with the limits in Eq.~\eqref{sRSK} if and only if $\beta=L^{-2}=\Lambda$.
	Under the hypotheses of Theorem~\ref{thm2}, this agreement is possible only for a spacelike-type end.
	
	\subsection{Energy conditions in the interior}
	
	In the orthonormal frame with timelike vector $\partial_\tau$, let $\rho_\tau$, $p_t$, and $p_\perp$ denote the effective energy density and principal pressures, each multiplied by $8\pi$.
	Here $p_t$ is the pressure along the spacelike $t$ direction, and $p_\perp$ is the pressure in either angular direction.
	Then
	\begin{equation}
		\begin{aligned}
			\rho_\tau &= k_4-2k_2, \qquad p_t=2k_3-k_4,\\
			p_\perp &= k_1+k_2+k_3.
		\end{aligned}
		\label{rhoptpp}
	\end{equation}
	The radial null energy condition is therefore
	$\rho_\tau+p_t=2(k_3-k_2)\geq0$.
	
	Since $\dot r<0$ just inside the future horizon, radial null completeness requires the radial null energy condition to fail somewhere in the interior.
	If $\dot r<0$ at times arbitrarily close to the inner end, the condition must also fail arbitrarily close to that end.
	Both conclusions hold for either type of inner
	end, as proved in Theorem~\ref{thm3} in Appendix~\ref{app1}.

	\subsection{Curvature along causal geodesics}
	
	Bounded scalar invariants do not in general imply bounded curvature in a parallel-propagated frame \cite{Zaslavskii2007,Bronnikov2008}.
	We therefore examine the components of the Riemann tensor in frames parallel transported along causal geodesics.
	
	The relevant curvature combination is
	\begin{equation}
		Q\equiv k_3-k_2=\frac{1}{2}(\rho_\tau+p_t).
		\label{eq:pp-Q}
	\end{equation}
	Suppose that $r\to L>0$ and the four $k_i$ are bounded near the inner end.
	By Theorem~\ref{thm4} in Appendix~\ref{app1}, all parallel-propagated curvature components are bounded along every timelike geodesic if and only if $Q/a^2$ is bounded.
	For radial null geodesics, both $Q/a^2$ and $a^2Q$ must be bounded.
	
	For nonradial null geodesics, one must also control the transport of the null frame.
	The theorem gives a sufficient integral condition that allows $r$ to oscillate.
	This condition holds automatically at a causally complete spacelike-type end.
	At such an end, parallel-propagated curvature is bounded along every timelike geodesic, and it is bounded along every null geodesic if and only if $a^2Q$ is bounded.
	
	\subsection{Cauchy surfaces}
	
	Along any future-directed causal curve approaching the interior	end,
	\begin{equation}
		\left|\frac{\dd t}{\dd\tau}\right|\leq\frac{1}{a}.
	\end{equation}
	At a spacelike-type end, $I_\star<\infty$, so $t$ has a finite limit.
	Together with $\zeta\to\zeta_i<0$, this gives finite limiting values of $T$ and $X$ along each such curve.
	If the interior is causally complete, timelike and null geodesics approach this end at infinite proper time and affine parameter, respectively.
	
	For a null-type domain, every surface $T=T_0$ is a Cauchy surface.
	For a spacelike-type domain, the constant-$T$ surfaces do not form a Cauchy foliation of the entire domain.
	Instead, define
	\begin{equation}
		w(T,X)=\frac{T}{\sqrt{X^2-\zeta_i}},
		\qquad -1<w<1.
		\label{eq:cauchy-w}
	\end{equation}
	Each level $w=w_0\in(-1,1)$ is a spacelike Cauchy surface.
	These statements are proved in Theorem~\ref{thm5} in Appendix~\ref{app1}.
	
	Both types of domain are therefore globally hyperbolic, with Cauchy surfaces of topology $\mathbb{R}\times S^2$.
	For the spacelike-type domain, the future inner end and future null infinities have limiting value $w=1$, while their past counterparts have $w=-1$.
	The causal diagram is shown in Fig.~\ref{f1}.
	
	\begin{figure}[htbp]      
		\centering             
		\includegraphics[width=0.45\textwidth]{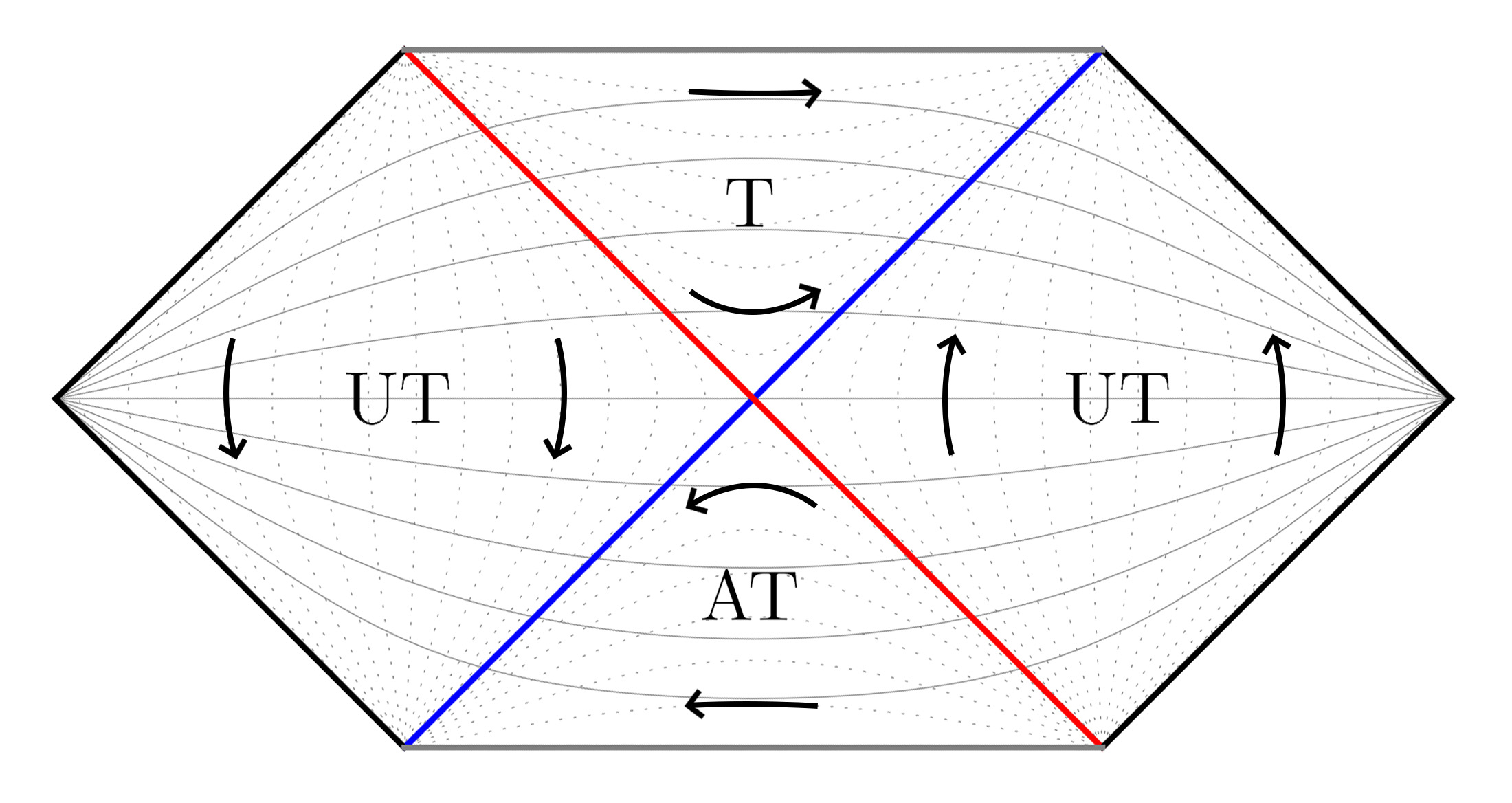}
		\caption{Causal diagram for $\zeta_i\in(-\infty,0)$, illustrated for $\dd r/\dd\zeta>0$.
			The regions are trapped ($\mathrm{T}$), anti-trapped ($\mathrm{AT}$), and untrapped ($\mathrm{UT}$).
			The blue and red lines mark the Killing horizons at $\zeta=0$.
			The four thick black lines represent null infinity, where $r\to\infty$, and the two thick gray lines represent the spacelike-type inner ends, where $r\to L$.
			The interior curves show constant-$r$ surfaces and the Cauchy surfaces $w=\mathrm{const}$.
			Arrows indicate the Killing field $\xi=\partial_t$.}
		\label{f1}    
	\end{figure}

	\subsection{A class of interior geometries}
	\label{sec36}
	
	We now consider interiors in which the areal radius decreases monotonically to $L>0$ on a final interval.
	Set $\delta=r-L$ and suppose that
	\begin{equation}
		D=-d_0\delta^p(1+u),\qquad
		\ee^\psi=A_0\delta^{-q}(1+v),
		\label{Depsi}
	\end{equation}
	where $d_0,A_0>0$.
	For each remainder $f=u,v$, we assume that
	\begin{equation}
		f\to0,\qquad \delta f'\to0,\qquad \delta^2 f''\to0
		\qquad (\delta\to0^+).
	\end{equation}
	The derivative assumptions allow us to evaluate the curvature from these expansions.
	
	The inner end is causally complete if and only if $p\geq2$ and $q\geq1$.
	Within this range, the causal type is determined by $q$:
	\begin{equation}
		\begin{aligned}
			1\leq q\leq p-1 &: \quad \text{null-type end},\\
			q>p-1 &: \quad \text{spacelike-type end}.
		\end{aligned}
		\label{eq:power-law-types}
	\end{equation}
	In both cases, $k_2,k_3\to0$, $k_4\to L^{-2}$, and $k_1\to-\beta$, where
	\begin{equation}
		\beta=
		\begin{cases}
			d_0(q-1)^2, & p=2,\\
			0,          & p>2.
		\end{cases}
		\label{eq:power-law-beta}
	\end{equation}
	Thus the scalar curvature invariants remain bounded throughout the complete range.
	For a null-type end, $\beta=0$: when $p=2$, the null-type condition requires $q=1$.
	
	Curvature in a parallel-propagated frame remains bounded along each fixed timelike geodesic throughout the complete range.
	It also remains bounded along each fixed null geodesic for every complete null-type end.
	For a complete spacelike-type end, it is bounded along all null geodesics if and only if
	\begin{equation}
		p-1<q\leq p-\frac{1}{2}.
		\label{pqp}
	\end{equation}
	When $q>p-1/2$, a parallel-propagated curvature component diverges along radial null geodesics, even though the scalar curvature invariants remain bounded.
	These results are derived in Appendix~\ref{app12}.
	
	The null-type and spacelike-type examples of Ref.~\cite{LiangLi2026AT} have $(p,q)=(4,3)$ and $(2,2)$, respectively.
	In the first example, parallel-propagated curvature is bounded along every causal geodesic.
	In the second, it is bounded along timelike geodesics but unbounded along radial null geodesics.
	
	\section{Energy conditions in the exterior}
	\label{sec4}
	
	\subsection{Asymptotic conditions}
	
	For the exterior analysis, we impose Schwarzschild asymptotics with control of the first two radial derivatives,
	\begin{equation}
		D(r)=1-\frac{2M}{r}+o_2(r^{-1}),
		\qquad
		\psi(r)=o_2(r^{-1}),
		\label{eq:exterior-asymptotics}
	\end{equation}
	where $M>0$.
	The notation $f=o_2(r^{-s})$ means
	\begin{equation}
		f^{(j)}(r)=o(r^{-s-j}),\qquad j=0,1,2.
	\end{equation}
	Here $f^{(j)}(r)$ denotes the $j$th derivative of $f$ with respect to $r$.
	These conditions control the metric and the derivatives that	enter the curvature and effective stress tensor.
	
	Writing $B=D\ee^{2\psi}=-g_{tt}$, Eq.~\eqref{eq:exterior-asymptotics}
	is equivalent to requiring
	\begin{equation}
		D(r)=1-\frac{2M}{r}+o_2(r^{-1}),
		\quad
		B(r)=1-\frac{2M}{r}+o_2(r^{-1}).
	\end{equation}
	Thus both metric functions have the same leading mass term.
	The equivalence is verified in Appendix~\ref{app21}.
	
	\subsection{Einstein tensor and energy inequalities}
	
	As in the interior, the effective density and pressures are rescaled by $8\pi$.
	In the static orthonormal frame, they are
	\begin{equation}
		\begin{aligned}
			\rho&=-G_t{}^t=\frac{1-D-rD'}{r^2},\\
			p_r&=G_r{}^r=-\rho+\frac{2D\psi'}{r},\\
			p_\perp&=G_\theta{}^\theta
			=\frac{D''}{2}+\frac{D'}{r}
			+D(\psi''+\psi'^2)\\
			&\quad+\frac{3}{2}D'\psi'+\frac{D\psi'}{r}.
		\end{aligned}
		\label{eq:exterior-stress}
	\end{equation}
	The effective stress tensor is diagonal in this frame and is of Hawking--Ellis type I.
	
	Define
	\begin{equation}
		\Sigma\equiv\rho+p_r=\frac{2D\psi'}{r}.
		\label{eq:exterior-Sigma}
	\end{equation}
	The null, weak, dominant, and strong energy conditions \cite{MartinMorunoVisser2017} hold simultaneously at an exterior
	point if and only if
	\begin{equation}
		\rho\geq0,\qquad
		0\leq\Sigma\leq2\rho,\qquad
		-\frac{\Sigma}{2}\leq p_\perp\leq\rho.
		\label{eq:exterior-all-EC}
	\end{equation}
	Their derivation is given in Appendix~\ref{app22}.
	
	\subsection{A class of exterior geometries}
	
	To obtain exteriors satisfying all four energy conditions at sufficiently large radii, choose
	\begin{equation}
		\begin{gathered}
			D(r)=1-\frac{2M}{r}+\mathcal P(r)
			+o_2(r^{-n-1}),\\
			\mathcal P(r)=\alpha\left(\frac Mr\right)^{n+1},
			\qquad \alpha>0,\quad n>0.
		\end{gathered}
		\label{eq:exterior-D}
	\end{equation}
	The correction is smaller than the Schwarzschild mass term, and the leading density is positive,
	\begin{equation}
		\rho=\rho_0+o(\rho_0),
		\qquad
		\rho_0=\frac{n\mathcal P}{r^2}>0.
	\end{equation}
	
	The radial energy inequalities and $\psi(\infty)=0$ imply $\psi\leq0$ and $\psi=O(\mathcal P)$.
	Their derivation is given in Appendix~\ref{app23}.
	For the construction, we choose the more specific expansion
	\begin{equation}
		\psi=-c_1\mathcal P+o_2(r^{-n-1}),
		\qquad
		c_1=\frac{nc_2}{2(n+1)},
		\label{eq:exterior-psi}
	\end{equation}
	where $c_2$ is constant.
	The resulting pressures satisfy
	\begin{equation}
		\begin{aligned}
			p_r&=(c_2-1)\rho_0+o(\rho_0),\\
			p_\perp&=\frac{n+1}{2}(1-c_2)\rho_0+o(\rho_0).
		\end{aligned}
		\label{eq:exterior-pressure-limits}
	\end{equation}
	The choice of $\psi$ affects whether the energy conditions hold.
	For example, $c_1=0$ gives $p_\perp/\rho\to(n+1)/2$, which violates the dominant energy condition when $n>1$.
	
	As proved in Theorem~\ref{thm6} in Appendix~\ref{app23}, the range
	\begin{equation}
		\max\left(0,\frac{n-1}{n+1}\right)
		<c_2<
		\min\left(2,\frac{n+1}{n}\right)
		\label{eq:tau-range}
	\end{equation}
	guarantees that all four energy conditions hold for every $r>R_{\mathrm{EC}}$, for some finite $R_{\mathrm{EC}}>r_h$.
	The choice $c_2=1$ lies in this interval for every $n>0$.
	When $c_2$ equals either endpoint of this interval, at least one inequality is saturated at leading order.
	The subleading terms must then be examined to determine whether all four energy conditions hold.
	
	\section{Global constructions}
	\label{sec5}
	
	We now construct metrics with a single simple horizon and a complete inner end of either type.
	In the first construction, the metric functions are real analytic for $r>L$, and all four standard energy conditions hold sufficiently far outside the horizon.
	In the second, the metric agrees exactly with Schwarzschild beyond a finite radius.
	
	Both constructions use the power-law interiors of Sec.~\ref{sec36}.
	For $p\geq2$, the complete inner ends fall into two ranges:
	\begin{equation}
		\begin{aligned}
			1\leq q\leq p-1 &: \quad \text{null-type end},\\
			q>p-1 &: \quad \text{spacelike-type end}.
		\end{aligned}
		\label{eq:construction-pq}
	\end{equation}
	The scalar curvature invariants are bounded in both ranges, as is the curvature in a parallel-propagated frame along each fixed timelike geodesic.
	Curvature is also bounded in such frames along all null geodesics if and only if $q\leq p-1/2$.
	This last condition holds throughout the null-type range.
	
	Appendix~\ref{app33} shows how both metrics extend across the horizon and recover the global form used in Sec.~\ref{sec2}.
	
	\subsection{Real-analytic metrics}
	\label{sec:analytic-construction}
	
	Fix $0<L<r_h$, a mass $M>0$, and a surface gravity $\kappa>0$.
	Write
	\begin{equation}
		D(r)=\left(1-\frac{r_h}{r}\right)\ee^{U(r)}.
	\end{equation}
	To obtain the exterior behavior of Sec.~\ref{sec4}, set
	\begin{equation}
		\mathcal P(r)=\alpha\left(\frac{M}{r}\right)^{n+1},
		\quad
		c_1=\frac{n c_2}{2(n+1)},
	\end{equation}
	where $\alpha,n>0$ and
	\begin{equation}
		\max\left(0,\frac{n-1}{n+1}\right)
		<c_2<
		\min\left(2,\frac{n+1}{n}\right).
		\label{eq:construction-c2}
	\end{equation}
	
	Choose integers $N>n$ and $m>n+1$, and define
	\begin{equation}
		\begin{aligned}
			S_N(r)
			&=\sum_{j=1}^{N}
			\frac{r_h^j-(2M)^j}{j r^j},\\
			H_N(r)
			&=\ln\frac{r}{r-L}
			-\sum_{j=1}^{N}\frac{L^j}{j r^j},\\
			F_h(r)
			&=\left(\frac{r_h}{r}\right)^m.
		\end{aligned}
		\label{eq:construction-functions}
	\end{equation}
	The three functions serve different purposes.
	The term $S_N$ fixes the asymptotic mass independently of $r_h$.
	The logarithm in $H_N$ gives the required powers of $r-L$ near the inner end, while the finite sum removes the terms that would alter the prescribed exterior expansion.
	Finally, $F_h(r_h)=1$ allows us to set the surface gravity without changing the leading behavior at either end.
	
	For any $\vartheta\in\mathbb R$, set
	\begin{equation}
		\begin{aligned}
			U
			&=S_N+\mathcal P-pH_N+(1-\vartheta)c_hF_h,\\
			\psi
			&=-c_1\mathcal P+qH_N+\vartheta c_hF_h,
		\end{aligned}
		\label{eq:analytic-metric}
	\end{equation}
	where
	\begin{equation}
		\begin{aligned}
			c_h={}&\ln(2\kappa r_h)-S_N(r_h)
			-(1-c_1)\mathcal P(r_h)\\
			&-(q-p)H_N(r_h).
		\end{aligned}
		\label{eq:analytic-ch}
	\end{equation}
	Then $U(r_h)+\psi(r_h)=\ln(2\kappa r_h)$, so the horizon has the prescribed surface gravity for every choice of $\vartheta$.
	
	The functions $D$ and $\psi$ are real analytic on $(L,\infty)$.
	Near $L$, they have the power laws in Eq.~\eqref{Depsi}, with positive amplitudes.
	At infinity, they satisfy
	\begin{equation}
		\begin{aligned}
			D&=1-\frac{2M}{r}+\mathcal P+o_2(r^{-n-1}),\\
			\psi&=-c_1\mathcal P+o_2(r^{-n-1}).
		\end{aligned}
		\label{eq:constructed-exterior}
	\end{equation}
	Thus either complete inner end in Eq.~\eqref{eq:construction-pq} can coexist with a regular horizon and a distant exterior satisfying the null, weak, dominant, and strong energy conditions.
	The verification is given in Appendix~\ref{app31}.
	
	\subsection{Smooth metrics with an exact Schwarzschild exterior}
	\label{sec:smooth-construction}
	
	We next require the metric to agree exactly with Schwarzschild beyond a finite radius.
	For this construction, take $r_h=2M$ and $\kappa=1/(4M)$.
	Choose $L<r_-<r_+$ and a smooth cutoff $\chi(r)$ such that
	\begin{equation}
		\chi(r)=
		\begin{cases}
			1, & r\leq r_-,\\
			0, & r\geq r_+,
		\end{cases}
		\qquad
		0<\chi(r)<1 \quad (r_-<r<r_+).
	\end{equation}
	Using $H_N$ from Eq.~\eqref{eq:construction-functions}, define
	\begin{equation}
		\begin{aligned}
			D&=\left(1-\frac{2M}{r}\right)
			\exp\!\left\{\chi\left[-pH_N+(1-\vartheta)c_s\right]\right\},\\
			\psi&=\chi\left(qH_N+\vartheta c_s\right),\\
			c_s&=(p-q)H_N(2M).
		\end{aligned}
		\label{eq:smooth-metric}
	\end{equation}
	This choice of $c_s$ preserves the Schwarzschild surface gravity regardless of the position of the transition interval relative to the horizon.
	
	For $L<r\leq r_-$, the cutoff equals one, and the exponents $p,q$ determine the inner end as in Eq.~\eqref{eq:construction-pq}.
	For $r\geq r_+$, the cutoff vanishes, so
	\begin{equation}
		D=1-\frac{2M}{r},\qquad \psi=0.
	\end{equation}
	If $r_+\leq2M$, the entire exterior is Schwarzschild.
	If $r_+>2M$, the exterior is Schwarzschild for $r\geq r_+$.
	The effective stress tensor vanishes in the Schwarzschild region, so all four energy conditions hold there.
	No energy condition is imposed in the transition region.
	Smooth matching and horizon regularity are verified in Appendix~\ref{app32}.
	
	The distinction between the two constructions is essential for exact matching.
	If $D$ and $\psi$ were real analytic on the connected interval $(L,\infty)$ and agreed with Schwarzschild on an open subinterval, the identity theorem would extend that agreement throughout $(L,\infty)$.
	The resulting metric would not have a complete inner end at $r=L$.

	\section{Applications to black hole interiors in effective loop quantum gravity}
	\label{sec6}
	
	B\"ohmer and Vandersloot (BV) studied the Schwarzschild interior using an improved effective Hamiltonian \cite{boehmer2007loop}.
	For $p_c>0$, their effective metric is
	\begin{equation}
		\dd s^2=-N(T)^2\dd T^2
		+\frac{p_b(T)^2}{L_0^2p_c(T)}\dd x^2
		+p_c(T)\dd\Omega^2,
		\label{eq:BV-metric}
	\end{equation}
	where $L_0$ is the fiducial coordinate length in the $x$ direction.
	In the time coordinate used by BV, the future inner end corresponds to $T\to-\infty$.
	Taking $N>0$ and choosing $\tau$ to increase toward the future, we identify the variables of Sec.~\ref{sec3} as
	\begin{equation}
		\dd\tau=-N\,\dd T,\quad
		r=\sqrt{p_c},\quad
		a=\frac{|p_b|}{L_0\sqrt{p_c}}.
		\label{eq:BV-identification}
	\end{equation}
	
	The numerical solution reported in Ref.~\cite{boehmer2007loop} approaches a Nariai-type solution as $T\to-\infty$, with
	\begin{equation}
		\begin{aligned}
			b(T)&\to\bar b, &
			p_b(T)&\sim\bar p_b\,\ee^{-\alpha T},\\
			c(T)&\sim\bar c\,\ee^{-\alpha T}, &
			p_c(T)&\to\bar p_c>0.
		\end{aligned}
		\label{eq:BV-asymptotics}
	\end{equation}
	Here $\alpha>0$, $\bar b$ and $\bar p_c$ are limiting values, and $\bar p_b$ and $\bar c$ are the coefficients of the leading exponentials.
	For the solution studied by BV, the lapse tends to a positive constant:
	\begin{equation}
		N(T)=\frac{\gamma\sqrt{\Delta}}
		{\sin\left(b(T)\sqrt{\Delta/p_c(T)}\right)}
		\longrightarrow \bar N>0.
		\label{eq:BV-lapse}
	\end{equation}
	Here $\gamma$ is the Barbero--Immirzi parameter and $\Delta$ is the area gap used in the effective Hamiltonian.
	Thus
	$r\to L=\sqrt{\bar p_c}$ and $a\sim a_0\ee^{-\alpha T}$, where $a_0=|\bar p_b|/(L_0L)>0$.
	
	For a fixed interior time $T_0$, the two integrals of Sec.~\ref{sec3} become
	\begin{equation}
		\begin{aligned}
			I_\star
			&=\int_{-\infty}^{T_0}\frac{N(T)}{a(T)}\,\dd T
			<\infty,\\
			I_{\mathrm{geo}}
			&=\int_{-\infty}^{T_0}
			\frac{N(T)}{\sqrt{1+a(T)^{-2}}}\,\dd T
			=\infty.
		\end{aligned}
		\label{eq:BV-integrals}
	\end{equation}
	The first integrand is asymptotic to
	$(\bar N/a_0)\ee^{\alpha T}$, whereas the second tends to $\bar N$.
	The BV inner end is therefore a spacelike-type
	asymptotic throat, and all causal geodesics directed toward it are complete.
	
	The sphere radius does not decrease monotonically toward this end.
	In the numerical solution, $p_c$ oscillates about $\bar p_c$ with decreasing amplitude \cite{boehmer2007loop}.
	Thus $L$ is the limiting radius, but it is not a lower bound on $r$.
	For the future radial null normals
	$\ell_\pm=\partial_\tau\pm a^{-1}\partial_x$, the expansions are
	\begin{equation}
		\theta_\pm=\frac{2}{r}\ell_\pm(r)=\frac{2\dot r}{r}.
		\label{eq:BV-expansions}
	\end{equation}
	The symmetry spheres are future trapped during contraction and anti-trapped during expansion.
	Both expansions vanish at each turning point.
	These changes do not affect the causal type of
	the inner end.
	
	The turning points also locate violations of the radial null energy condition.
	With the normalization of Sec.~\ref{sec3},
	\begin{equation}
		\rho_\tau+p_t=2Q,\quad
		Q=k_3-k_2=\frac{H_a\dot r-\ddot r}{r},
		\quad H_a=\frac{\dot a}{a}.
	\end{equation}
	At a turning point, $Q=-\ddot r/r$.
	Thus the radial null energy condition fails at each nondegenerate local minimum of $r$ and
	holds strictly at each nondegenerate local maximum.
	Both occur arbitrarily close to the inner end, as shown in Appendix~\ref{app:BV-curvature}.
	
	On the constant-radius Nariai-type solution, $r=L$ and $a\propto\ee^{H\tau}$, where $H=\alpha/\bar N$.
	Its curvature functions are
	\begin{equation}
		k_1=-H^2,\qquad k_2=k_3=0,\qquad k_4=L^{-2}.
		\label{eq:BV-background-curvature}
	\end{equation}
	Along the oscillatory trajectory, the four curvature functions tend to these same values (Appendix~\ref{app:BV-curvature}).
	The scalar curvature invariants of Sec.~\ref{sec3} therefore have finite limits with $\beta=H^2$.
	Since $Q\to0$ and $a\to\infty$, we also have $Q/a^2\to0$.
	Curvature in a parallel-propagated frame along each fixed timelike geodesic is consequently bounded near the end.
	
	Radial null geodesics distinguish the two solutions.
	On the constant-radius solution, $Q=0$ identically.
	On the oscillatory trajectory, the radial null criterion also involves $a^2Q$.
	Although $Q\to0$, its decay at the turning points is too slow to compensate for the growth of $a^2$.
	Using the full BV effective equations, Appendix~\ref{app:BV-curvature} shows that
	at successive turning points $T_n\to-\infty$,
	\begin{equation}
		|a^2Q|(T_n)\longrightarrow\infty.
		\label{eq:BV-a2Q-unbounded}
	\end{equation}
	Thus a Riemann component in a parallel-propagated frame along a radial null geodesic is unbounded, although the geodesic has infinite affine parameter and the scalar curvature invariants remain bounded.
	
	\section{Discussion and conclusions}
	\label{sec7}
	
	A spherically symmetric black hole with a single simple horizon can have a causally complete inner end at a finite, positive limiting areal radius $L$.
	Our explicit constructions realize both null-type and spacelike-type ends with a radius that decreases throughout the interior.
	In both cases, the spacetime admits Cauchy surfaces of topology $\mathbb{R}\times S^2$.
	In this sense, an asymptotic throat provides a minimal geometric alternative to the Schwarzschild singularity, requiring neither a bounce nor a change in spatial topology.
	
	Keeping two independent metric functions, we obtained two integral criteria that also apply when the radius approaches $L$ through oscillations.
	The integral $I_\star$ determines the causal character
	of the inner end.
	The integral $I_{\mathrm{geo}}$ diverges if and only if all causal geodesics approaching the end have infinite proper time or affine parameter.
	
	The two types of complete end admit different curvature limits.
	Suppose that the Kretschmann scalar is bounded and that $k_1$ and $k_3$ have finite limits.
	At a null-type end, $L$ fixes the limiting values of $R$, $\mathcal{S}$, and $\mathcal{K}$.
	At a spacelike-type end, these values also depend on
	$\beta=-\lim k_1\geq0$.
	For $\beta=L^{-2}$, the three limits agree with the scalar invariants of Nariai spacetime.
	
	Suppose that the Kretschmann scalar is bounded near the inner end, and let $Q=k_3-k_2$.
	Curvature in a parallel-propagated frame is bounded along every timelike geodesic approaching the end if and
	only if $Q/a^2$ remains bounded.
	For radial null geodesics, the corresponding criterion requires both $Q/a^2$ and $a^2Q$ to remain bounded.
	Even at a complete inner end, a parallel-propagated curvature component can grow without bound at infinite affine parameter.
	
	The B\"ohmer--Vandersloot model \cite{boehmer2007loop} illustrates these results in an effective loop quantum gravity interior.
	The oscillatory solution approaches a complete spacelike-type end.
	As the radius oscillates about $L$, the symmetry spheres alternate between future trapped and anti-trapped.
	The scalar curvature invariants approach those of the constant-radius Nariai-type solution.
	Along each radial null geodesic, however, a parallel-propagated curvature component is unbounded along a sequence of turning points as the affine parameter tends to infinity.
	
	The construction leaves $L$ as a geometric parameter.
	A matter model or a microscopic theory would be needed to determine what sets this scale.
	Whether such throats form in gravitational collapse and whether they are stable under perturbations remain open questions.
	The behavior of covariant derivatives of curvature also remains to be studied.
	
	\begin{acknowledgments}
		The authors used GPT-6.2 Astra (OpenAI) to help check the proofs and derivations.
		They take full responsibility for the arguments, results, and conclusions presented in this work.
	\end{acknowledgments}
	
	\appendix
	\section{Geometric results for the black hole interior}
	\label{app1}
	
	\subsection{General results}
	\label{app11}
	
	\begin{theorem}[]
		\label{thm1}
		Consider the future interior described by Eqs.~\eqref{taur} and~\eqref{taut}, with spatial coordinates ranging over $\mathbb{R}\times S^2$.
		Assume that $a$ and $r$ are smooth and positive, and that $r(\tau)\to L>0$ as $\tau\to\tau_i$.
		Fix $\tau_0$ strictly inside this region, and define $I_\star$ and $I_{\mathrm{geo}}$ by Eqs.~\eqref{Istar} and~\eqref{Igeo}.
		Then:
		\begin{enumerate}
			\item[(i)] The inner end is of null type if and only if $I_\star=\infty$, and of spacelike type if and only if $I_\star<\infty$.
			\item[(ii)] All future-directed causal geodesics approaching the inner end have infinite proper time or affine length if and only if $I_{\mathrm{geo}}=\infty$.
			\item[(iii)] The inner end is of spacelike type and causally complete if and only if $I_\star<\infty$ and $\tau_i=\infty$.
		\end{enumerate}
	\end{theorem}
	
	\par\noindent
	\textit{Proof.}
	First, Eq.~\eqref{taur} gives
	\begin{equation}
		\frac{\dd\ln(-\zeta)}{\dd\tau}
		=\frac{2\kappa}{a}.
	\end{equation}
	Writing $\zeta_0=\zeta(\tau_0)$, integration yields
	\begin{equation}
		\ln\frac{-\zeta(\tau)}{-\zeta_0}
		=2\kappa\int_{\tau_0}^{\tau}
		\frac{\dd s}{a(s)}.
		\label{lnzeta}
	\end{equation}
	Hence
	\begin{equation}
		I_\star
		=\frac{1}{2\kappa}
		\lim_{\zeta\to\zeta_i}
		\ln\frac{-\zeta}{-\zeta_0}.
	\end{equation}
	This limit diverges precisely when $\zeta_i=-\infty$ and is finite when $-\infty<\zeta_i<0$.
	By the definition of the two types of inner end, this proves (i).
	
	Next, spherical symmetry allows each geodesic to be placed in the equatorial plane $\theta=\pi/2$ \cite{Chandrasekhar1992}.
	Its Lagrangian is
	\begin{equation}
		\begin{aligned}
			2\mathcal{L}={}&
			-\left(\frac{\dd\tau}{\dd\lambda}\right)^2
			+a^2\left(\frac{\dd t}{\dd\lambda}\right)^2\
			&+r^2\left(\frac{\dd\varphi}{\dd\lambda}\right)^2
			=\epsilon,
		\end{aligned}
		\label{Lagrangian}
	\end{equation}
	where $\epsilon=-1$ for timelike geodesics parametrized by proper time and $\epsilon=0$ for affinely parametrized null geodesics.
	The conserved quantities are
	\begin{equation}
		\Pi=a^2\frac{\dd t}{\dd\lambda},
		\qquad
		J=r^2\frac{\dd\varphi}{\dd\lambda}.
	\end{equation}
	Substitution into Eq.~\eqref{Lagrangian} gives
	\begin{equation}
		\left(\frac{\dd\tau}{\dd\lambda}\right)^2
		=-\epsilon+\frac{\Pi^2}{a^2}+\frac{J^2}{r^2}.
	\end{equation}
	For null geodesics, $\Pi$ and $J$ cannot both vanish.
	For future-directed geodesics, $\dd\tau/\dd\lambda>0$.
	The parameter interval from $\tau_0$ to the inner end is therefore
	\begin{equation}
		\Delta\lambda
		=\int_{\tau_0}^{\tau_i}
		\frac{\dd\tau}
		{\sqrt{-\epsilon+\Pi^2/a^2+J^2/r^2}}.
		\label{Deltalambda}
	\end{equation}
	Since $r$ is positive and tends to $L>0$, there is a constant $m>0$ such that $r(\tau)\ge m$ on $[\tau_0,\tau_i)$. Define
	\begin{equation}
		C_{\mathrm{geo}}
		=\max\left\{
			1,\Pi^2,-\epsilon+\frac{J^2}{m^2}
			\right\}.
	\end{equation}
	Then
	\begin{equation}
		-\epsilon+\frac{\Pi^2}{a^2}+\frac{J^2}{r^2}
		\le C_{\mathrm{geo}}(1+a^{-2}),
	\end{equation}
	and therefore
	\begin{equation}
		\Delta\lambda
		\ge C_{\mathrm{geo}}^{-1/2}I_{\mathrm{geo}}.
	\end{equation}
	Thus $I_{\mathrm{geo}}=\infty$ guarantees completeness toward the inner end for every causal geodesic.
	Conversely, a radial timelike geodesic with $\Pi=1$ and $J=0$ has $\Delta\lambda=I_{\mathrm{geo}}$.
	Completeness of all such geodesics therefore requires $I_{\mathrm{geo}}=\infty$. This proves (ii).
	
	Finally, the integrand defining $I_{\mathrm{geo}}$ is at most one, so completeness requires $\tau_i=\infty$.
	Suppose now that $I_\star<\infty$ and $\tau_i=\infty$. For $a>0$,
	\begin{equation}
		\frac{1}{\sqrt{1+a^{-2}}}
		\ge 1-\frac{1}{a}.
	\end{equation}
	Consequently, for every finite $\tau>\tau_0$,
	\begin{equation}
			\int_{\tau_0}^{\tau}
			\frac{\dd s}{\sqrt{1+a^{-2}(s)}}
			\ge \tau-\tau_0
			-\int_{\tau_0}^{\tau}\frac{\dd s}{a(s)}\
			\ge \tau-\tau_0-I_\star.
	\end{equation}
	Letting $\tau\to\infty$ gives $I_{\mathrm{geo}}=\infty$.
	Together with (i) and (ii), this proves (iii).
	
	Thus $\tau_i=\infty$ is necessary for causal completeness for both types of inner end.
	It is also sufficient for spacelike-type ends.
	For null-type ends, the condition $I_{\mathrm{geo}}=\infty$ must still be checked.
	\hfill$\square$
	
	\begin{theorem}[]
		\label{thm2}
		Suppose the interior metric \eqref{taut} is causally complete toward the inner end, $r(\tau)\to L>0$, and $\mathcal{K}$ is bounded near the end.
		Then $H_a\equiv\dot a/a$ is bounded near the end, and
		\begin{equation}
			\dot r\to0,\qquad k_2\to0,\qquad k_4\to L^{-2}.
			\label{dotrH}
		\end{equation}
		If the end is spacelike-type, then also
		\begin{equation}
			a\to\infty,\qquad A\to\infty.
			\label{aA}
		\end{equation}
		
		Suppose in addition that $k_1$ and $k_3$ have finite limits.
		Then there is a constant $\beta\geq0$ such that
		\begin{equation}
			k_1\to-\beta,\qquad k_3\to0,
			\label{sk1234}
		\end{equation}
		and the scalar limits are given by Eq.~\eqref{sRSK}.
		If the end is null-type, then $\beta=0$.
	\end{theorem}
	
	\noindent\textit{Proof.}
	Causal completeness gives $\tau_i=\infty$ by Theorem~\ref{thm1}.
	Equation~\eqref{eq:curvature-scalars} shows that all four $k_i$ are bounded on some interval $[T,\infty)$.
	Since $\ddot r=-rk_3$ and $r\to L$, there is a constant $C_r$ such that $|\ddot r|\leq C_r$ on this interval.
	For each $h>0$, Taylor's theorem gives
	\begin{equation}
		|\dot r(\tau)|
		\leq \frac{|r(\tau+h)-r(\tau)|}{h}+\frac{C_r h}{2}.
	\end{equation}
	For fixed $h$, the first term tends to zero because both radius values tend to $L$.
	Taking the upper limit as $\tau\to\infty$ and then letting $h\to0$ gives $\dot r\to0$.
	It follows that $k_4\to L^{-2}$.
	
	Choose $C>0$ such that $|k_1|\leq C$ on $[T,\infty)$.
	Then
	\begin{equation}
		\dot H_a=-k_1-H_a^2\leq C-H_a^2.
	\end{equation}
	Since $\dot H_a<0$ whenever $H_a>\sqrt C$, we have
	\begin{equation}
		H_a(\tau)\leq\max\{H_a(T),\sqrt C\}
		\qquad (\tau\geq T).
	\end{equation}
	To obtain a lower bound, suppose that $H_a(\tau_1)<-\sqrt C$ at some $\tau_1\geq T$. Set $u=-H_a-\sqrt C$.
	Then $u(\tau_1)>0$ and
	\begin{equation}
		\dot u\geq H_a^2-C=u^2+2\sqrt C\,u.
	\end{equation}
	Thus $u$ increases while it is positive, so it remains positive for $\tau\geq\tau_1$.
	It follows that $\dot u\geq u^2$. Integrating
	$\dd(1/u)/\dd\tau\leq-1$ gives
	\begin{equation}
		0<\frac{1}{u(\tau)}
		\leq\frac{1}{u(\tau_1)}-(\tau-\tau_1).
	\end{equation}
	The right-hand side vanishes at the finite time $\tau=\tau_1+1/u(\tau_1)$.
	This is a contradiction, since $\tau_i=\infty$ and $a$ is smooth and positive at every finite $\tau$, so $H_a=\dot a/a$ and $u$ are finite there.
	Hence $H_a\geq-\sqrt C$ on $[T,\infty)$.
	We have therefore bounded $H_a$ from above and below.
	Since $\dot r\to0$ and $r\to L>0$, we obtain $k_2=-H_a\dot r/r\to0$.
	
	For a spacelike-type end, $I_\star<\infty$.
	Let $|H_a|\leq B$ near the end and fix $h>0$.
	Integrating $H_a=\dot a/a$ gives $a(s)\leq a(\tau)\ee^{Bh}$ for $\tau\leq s\leq\tau+h$.
	Hence, for fixed $h>0$ and sufficiently large $\tau$,
	\begin{equation}
		0<\frac{1}{a(\tau)}
		\leq \frac{\ee^{Bh}}{h}
		\int_\tau^{\tau+h}\frac{\dd s}{a(s)}
		\leq \frac{\ee^{Bh}}{h}
		\int_\tau^\infty\frac{\dd s}{a(s)}.
	\end{equation}
	Since $I_\star<\infty$, the last integral tends to zero as $\tau\to\infty$.
	Therefore $a(\tau)\to\infty$ as $\tau\to\infty$.
	Since $A=a^2/(\kappa^2|\zeta|)$ and $\zeta\to\zeta_i\in(-\infty,0)$, we also have $A\to\infty$.
	
	Now suppose that $k_1$ and $k_3$ have finite limits.
	If $k_3\to\ell\ne0$, then $\ddot r\to-L\ell\ne0$, which contradicts $\dot r\to0$.
	Thus $k_3\to0$.
	If $k_1$ had a positive limit, then $k_1\geq c>0$ for all sufficiently large $\tau$.
	The identity $\dot H_a=-k_1-H_a^2$ would give $\dot H_a\leq-c$, contradicting the lower bound on $H_a$.
	We may therefore write $k_1\to-\beta$ with $\beta\geq0$.
	
	It remains to show that $\beta=0$ for a null-type end.
	Suppose $\beta>0$.
	Then $\ddot a\geq ca$ on some interval $[T,\infty)$, with $c>0$, so $\dot a$ is increasing.
	If $\dot a(T_1)\geq0$ at any $T_1\geq T$, then $a(\tau)\geq a(T_1)$ and $\ddot a(\tau)\geq ca(T_1)$ for $\tau\geq T_1$.
	Integrating twice gives
	\begin{equation}
		a(\tau)\geq a(T_1)
		\left[1+\frac{c}{2}(\tau-T_1)^2\right].
	\end{equation}
	This makes $I_\star$ finite, contrary to the null-type condition.
	The other possibility is that $\dot a<0$ on $[T,\infty)$.
	For every $S>T$,
	\begin{equation}
		c\int_T^S a(\tau)\,\dd\tau
		\leq \dot a(S)-\dot a(T)
		\leq -\dot a(T).
	\end{equation}
	Thus $\int_T^\infty a(\tau)\,\dd\tau$ is finite.
	Since $(1+a^{-2})^{-1/2}\leq a$, this makes $I_{\mathrm{geo}}$ finite, contrary to causal completeness.
	Hence $\beta=0$ for a null-type end.
	Substitution into Eq.~\eqref{eq:curvature-scalars} gives Eq.~\eqref{sRSK}.
	\hfill$\square$
	
	\begin{theorem}[]
		\label{thm3}
		Suppose that $r(\tau)\to L>0$ as $\tau\to\tau_i$ and that future-directed radial null geodesics are complete toward the inner end.
		For every $T<\tau_i$ at which $\dot r(T)<0$, there is a later time $\tau\in(T,\tau_i)$ at which $\rho_\tau+p_t<0$.
		
		In particular, if $\dot r<0$ at times arbitrarily close to $\tau_i$, then the radial null energy condition fails at points arbitrarily close to $\tau_i$.
		This result holds for both null-type and spacelike-type ends.
	\end{theorem}
	
	\noindent\textit{Proof.}
	Set $z=\dot r/a$. By Eq.~\eqref{rhoptpp} and the definitions of $k_2$ and $k_3$,
	\begin{equation}
		\dot z
		=\frac{\ddot r-H_a\dot r}{a}
		=-\frac{r}{a}(k_3-k_2)
		=-\frac{r}{2a}(\rho_\tau+p_t).
	\end{equation}
	Fix a time $T$ with $\dot r(T)<0$, so that $z(T)<0$.
	Suppose the radial null energy condition holds throughout $(T,\tau_i)$.
	Then $\dot z\leq0$, and hence $z(\tau)\leq z(T)<0$ for $\tau>T$.
	Integrating $\dot r(\tau)=a(\tau)z(\tau)\leq a(\tau)z(T)$ gives
	\begin{equation}
		r(T)-r(S)
		\geq -z(T)\int_T^S a(\tau)\,\dd\tau,
		\quad T<S<\tau_i.
	\end{equation}
	
	Along a future-directed radial null geodesic, $\dd\tau/\dd\lambda=|\Pi|/a$, where $\lambda$ is an affine parameter and $\Pi\ne0$.
	Thus
	\begin{equation}
		\lambda(S)-\lambda(T)
		=\frac{1}{|\Pi|}\int_T^S a(\tau)\,\dd\tau.
	\end{equation}
	By radial null completeness, this integral diverges as $S\to\tau_i$.
	Since $-z(T)>0$, the right-hand side of the preceding inequality tends to infinity, whereas its left-hand side tends to the finite value $r(T)-L$.
	This contradiction proves the first claim. If contracting times $T$ occur arbitrarily close to $\tau_i$, each is followed by a violation, which proves the second claim.
	\hfill$\square$
	
	\begin{theorem}[]
		\label{thm4}
		Suppose that $r(\tau)\to L>0$ and that the four curvature functions $k_i$ are bounded near the inner end.
		The following statements concern future-directed geodesics approaching that end.
		The bound on the Riemann components may depend on the chosen geodesic and initial frame.
		
		\begin{enumerate}
			\item Along every timelike geodesic, all Riemann components are bounded in a parallel-propagated rest frame if and only if $Q/a^2$ is bounded near the end.
			
			\item Along every radial null geodesic, all Riemann components are bounded in a parallel-propagated null frame
			if and only if both $Q/a^2$ and $a^2Q$ are bounded near the end.
			
			\item Suppose, in addition, that for some $\tau_0<\tau_i$,
			\begin{equation}
				\int_{\tau_0}^{\tau_i}
				\frac{|\dot r|}{\sqrt{1+a^2}}\,\dd\tau<\infty.
				\label{eq:null-frame-integral}
			\end{equation}
			Then all Riemann components are bounded in a parallel-propagated null frame along every null geodesic if and only if both $Q/a^2$ and $a^2Q$ are bounded.
			
			\item At a causally complete spacelike-type end, Eq.~\eqref{eq:null-frame-integral} holds automatically.
			All parallel-propagated Riemann components are then bounded along every timelike geodesic. They are bounded
			along every null geodesic if and only if $a^2Q$ is bounded.
		\end{enumerate}
	\end{theorem}
	
	\noindent\textit{Proof.}
	Use the orthonormal frame
	\begin{equation}
		\begin{gathered}
			e_0=\partial_\tau,\quad e_1=a^{-1}\partial_t,\\
			e_2=r^{-1}\partial_\theta,\quad
			e_3=(r\sin\theta)^{-1}\partial_\varphi.
		\end{gathered}
	\end{equation}
	Its independent curvature components are
	\begin{equation}
		\begin{gathered}
			R_{0101}=k_1,\quad R_{0m0m}=k_3,\\
			R_{1m1m}=-k_2,\quad R_{2323}=k_4,
			\quad m=2,3.
		\end{gathered}
		\label{Re0123}
	\end{equation}
	By spherical symmetry, we may place each geodesic in the equatorial plane.
	Write $w=\Pi/a$ and $j=J/r$.
	Since $r\to L>0$, $r$ is bounded above and bounded away from zero near the end.
	
	First consider a unit timelike tangent $u=\gamma e_0+w e_1+j e_3$, where $\gamma^2=1+w^2+j^2$.
	Set $b=\sqrt{1+j^2}$, $W=w/b$, and $\Gamma=\gamma/b$.
	Then $\Gamma^2-W^2=1$.
	The radial boost
	\begin{equation}
		\begin{gathered}
			f_0=\Gamma e_0+W e_1,\quad
			f_1=W e_0+\Gamma e_1,\\
			f_2=e_2,\quad f_3=e_3
		\end{gathered}
		\label{f0123}
	\end{equation}
	leaves $R_{0101}$ and $R_{2323}$ unchanged.
	The other independent components become
	\begin{equation}
		\begin{gathered}
			R^{(f)}_{0m0m}=k_3+W^2Q,\\
			R^{(f)}_{1m1m}=-k_2+W^2Q,\\
			R^{(f)}_{0m1m}=\Gamma WQ.
		\end{gathered}
		\label{Rf0123}
	\end{equation}
	If $Q/a^2$ is bounded, these components are bounded because $W^2\leq\Pi^2/a^2$ and $|\Gamma W|\leq1+W^2$.
	A second boost sends $f_0$ to $u=b f_0+j f_3$.
	Its coefficients $b$ and $j$ are bounded along the fixed geodesic.
	Thus all components are bounded in this rest frame and in every other rest frame, since two rest frames with time axis $u$ differ only by a spatial rotation.
	Conversely, take a radial timelike geodesic with $J=0$ and $\Pi\ne0$.
	Then $f_0=u$ and
	\begin{equation}
		R(u,e_2,u,e_2)=k_3+\Pi^2Q/a^2.
	\end{equation}
	Since $k_3$ is bounded, $Q/a^2$ must be bounded. This proves part~1.
	
	For a radial null geodesic, let $E=|\Pi|>0$ and $\sigma=\operatorname{sgn}\Pi$.
	A parallel-propagated null frame is
	\begin{equation}
		k=\frac{E}{a}(e_0+\sigma e_1),\quad
		l=\frac{a}{2E}(e_0-\sigma e_1),\quad e_2,\quad e_3,
	\end{equation}
	where $k$ is the affine tangent and $g(k,l)=-1$.
	In this frame,
	\begin{equation}
		\begin{gathered}
			R_{kmkm}=\frac{E^2}{a^2}Q,\quad 
			R_{lmlm}=\frac{a^2}{4E^2}Q,\\
			R_{kmlm}=\frac{k_2+k_3}{2}, \quad 
			R_{klkl}=k_1,\quad R_{2323}=k_4 .
		\end{gathered}
	\end{equation}
	This proves part~2 and gives the necessary conditions in part~3.
	
	Now consider a nonradial null geodesic.
	Its affine tangent is $k=\gamma e_0+w e_1+j e_3$, with $\gamma^2=w^2+j^2$ and
	$J\ne0$.
	Use the boost \eqref{f0123} with $b=|j|$, $W=w/b$, and $\Gamma=\gamma/b$.
	Both $b$ and $b^{-1}$ are bounded near the end.
	Hence Eq.~\eqref{Rf0123} shows that bounded $Q/a^2$ gives bounded curvature in the boosted frame.
	With $\sigma=\operatorname{sgn}J$, the same is true in the normalized null frame
	\begin{equation}
		k=b(f_0+\sigma f_3),\quad
		l_0=\frac{f_0-\sigma f_3}{2b},\quad f_1,\quad f_2.
	\end{equation}
	To obtain a parallel-propagated frame, write $H_r=\dot r/r$.
	Differentiation along the geodesic gives
	\begin{equation}
		\nabla_k f_1=H_rWk,\quad
		\nabla_k l_0=H_rW f_1,\quad
		\nabla_k f_2=0.
	\end{equation}
	Set $z(\tau_0)=0$ and $\dd z/\dd\lambda=-H_rW$. Then
	\begin{equation}
		m_1=f_1+zk,\quad
		l=l_0+zf_1+\frac{z^2}{2}k
	\end{equation}
	give a parallel-propagated frame $\{k,l,m_1,f_2\}$.
	Since $\dd\tau/\dd\lambda=\gamma$,
	\begin{equation}
		\frac{\dd z}{\dd\tau}
		=-\frac{\Pi\dot r}
		{|J|\sqrt{\Pi^2+a^2J^2/r^2}}.
		\label{eq:null-frame-z}
	\end{equation}
	For $\Pi=0$, $z=0$. For $\Pi\ne0$, choose an upper bound $M$ for $r$ near the end.
	Then
	\[
	\Pi^2+a^2J^2/r^2
	\geq \min\{\Pi^2,J^2/M^2\}(1+a^2).
	\]
	It follows from Eqs.~\eqref{eq:null-frame-integral} and~\eqref{eq:null-frame-z} that $z$ is bounded.
	The null rotation therefore preserves boundedness of all Riemann components.
	Any two parallel-propagated frames along the same geodesic differ by a constant transformation, so the result holds for every fixed initial frame.
	This proves part~3.
	
	Finally, at a causally complete spacelike-type end, Theorem~\ref{thm2} gives $a\to\infty$ and $\dot r\to0$.
	In particular, $Q/a^2\to0$ because $Q$ is bounded.
	Also, $|\dot r|\leq C$ on a final interval, so
	\[
	\int_{\tau_0}^{\tau_i}
	\frac{|\dot r|}{\sqrt{1+a^2}}\,\dd\tau
	\leq C\int_{\tau_0}^{\tau_i}\frac{\dd\tau}{a}
	<\infty,
	\]
	where $\tau_0$ is chosen in that interval.
	The last inequality is the spacelike-type condition $I_\star<\infty$.
	Parts~1 and~3 now give part~4.
	\hfill$\square$
	
	\begin{theorem}[]
		\label{thm5}
		Consider the global metric \eqref{globalmetric} on
		\begin{equation}
			\Omega_i=
			\{(T,X):X^2-T^2>\zeta_i\}\times S^2,
		\end{equation}
		where $A$ and $r$ are smooth and positive throughout their domain.
		
		If $\zeta_i=-\infty$, every surface $T=T_0$, with $T_0\in\mathbb{R}$, is a spacelike Cauchy surface.
		
		If $-\infty<\zeta_i<0$, every level $w=w_0$, with $w$ defined by Eq.~\eqref{eq:cauchy-w} and $w_0\in(-1,1)$, is a spacelike Cauchy surface.
		
		In both cases, these surfaces have topology $\mathbb{R}\times S^2$, and the spacetime is globally hyperbolic.
	\end{theorem}
	
	\noindent\textit{Proof.}
	Choose the time orientation in which $\partial_T$ is future-directed.
	Since $g^{TT}=-1/A<0$, $T$ increases along every future-directed causal curve.
	Using $T$ as a parameter and writing $\omega(T)\in S^2$ for the angular position, causality gives
	\begin{equation}
		\left|\frac{\dd X}{\dd T}\right|\leq1,
		\qquad
		\left|\frac{\dd\omega}{\dd T}\right|_{S^2}
		\leq\frac{\sqrt A}{r}.
		\label{eq:global-causal-bounds}
	\end{equation}
	
	We first consider a future-inextendible causal curve for which $T$ has a finite upper limit $T_+$.
	The first inequality implies that $X(T)$ has a finite limit $X_+$.
	Set $\zeta_+=X_+^2-T_+^2$.
	If $\zeta_+>\zeta_i$, smoothness and positivity give a finite upper bound for $\sqrt A/r$ near the limiting point.
	The angular speed is therefore bounded by some constant $C$.
	For $T_1<T_2$ sufficiently close to $T_+$, the spherical distance satisfies
	\begin{equation}
		d_{S^2}\bigl(\omega(T_1),\omega(T_2)\bigr)
		\leq C(T_2-T_1).
	\end{equation}
	Since $S^2$ is complete, $\omega(T)$ has a limit $\omega_+$.
	Thus the causal curve has an endpoint $(T_+,X_+,\omega_+)$ in $\Omega_i$ and can be extended, contrary to future inextendibility.
	Consequently, a finite upper limit of $T$ is possible only when $\zeta_i$ is finite and $\zeta_+=\zeta_i$.
	The corresponding statement holds in the past.
	
	For $\zeta_i=-\infty$, this argument shows that $T$ ranges from $-\infty$ to $\infty$ along every inextendible causal curve.
	Each surface $T=T_0$ is therefore crossed exactly once.
	It is spacelike because $\nabla T$ is timelike.
	
	Now suppose $\zeta_i=-c$, with $c>0$.
	Then
	\begin{equation}
		w=\frac{T}{\sqrt{X^2+c}},\qquad -1<w<1.
	\end{equation}
	Its gradient satisfies
	\begin{equation}
		g^{ab}\nabla_a w\nabla_b w
		=-\frac{c+(1-w^2)X^2}{A(X^2+c)^2}<0,
	\end{equation}
	so every level of $w$ is spacelike.
	Along a future-directed causal curve,
	\begin{equation}
		\frac{\dd w}{\dd T}
		=\frac{X^2+c-TX\,\dd X/\dd T}{(X^2+c)^{3/2}}>0,
	\end{equation}
	because $|\dd X/\dd T|\leq1$ and $|TX|<X^2+c$.
	Thus each level is crossed at most once.
	
	To prove that every level is crossed, consider the future end of an inextendible causal curve.
	If $T_+<\infty$, the preceding argument gives $X_+^2-T_+^2=-c$.
	Hence the limiting value of $w$ is either $1$ or $-1$.
	Since $w$ is increasing and takes values in $(-1,1)$, its future limit must be $1$.
	
	If instead $T\to\infty$, the first inequality in Eq.~\eqref{eq:global-causal-bounds} gives
	$|X(T)|\leq T+C$ for some constant $C\geq0$ and all sufficiently large $T$.
	Therefore
	\begin{equation}
		\frac{T}{\sqrt{(T+C)^2+c}}
		\leq w(T)<1,
	\end{equation}
	which implies $w\to1$ as $T\to\infty$.
	Applying the same argument with the time orientation reversed gives $w\to-1$ at the past end.
	
	Thus every inextendible causal curve crosses each level $w=w_0\in(-1,1)$ exactly once.
	These levels are Cauchy surfaces.
	They are the graphs
	\begin{equation}
		T=w_0\sqrt{X^2+c},
		\qquad (X,\omega)\in\mathbb{R}\times S^2,
	\end{equation}
	and therefore have topology $\mathbb{R}\times S^2$.
	The same topology holds for the constant-$T$ surfaces in the
	null-type case.
	\hfill$\square$
	
	\subsection{Derivation of the power-law criteria}
	\label{app12}
	
	On the final interval, $\dot r=-\sqrt{-D}$ and $a=\sqrt{-D}\,\ee^\psi$.
	Changing the integration variable from $\tau$ to $\delta=r-L$, we obtain
	\begin{equation}
		I_\star=\int_0^{\delta_0}
		\frac{\dd\delta}{(-D)\ee^\psi},
		\quad
		I_{\mathrm{geo}}=\int_0^{\delta_0}
		\frac{\dd\delta}{\sqrt{\ee^{-2\psi}-D}},
	\end{equation}
	where $\delta_0>0$ lies in this interval.
	By Eq.~\eqref{Depsi}, the integrand of $I_\star$ is asymptotic to $(d_0A_0)^{-1}\delta^{q-p}$.
	Hence $I_\star$ diverges if and only if $q\leq p-1$, including logarithmic divergence at $q=p-1$.
	
	The integrand of $I_{\mathrm{geo}}$ is asymptotic to a positive constant times $\delta^{-\min(q,p/2)}$.
	Thus $I_{\mathrm{geo}}$ diverges if and only if $q\geq1$ and $p\geq2$.
	Combining the two integral criteria gives Eq.~\eqref{eq:power-law-types}.
	
	Substitution of Eq.~\eqref{Depsi} into the curvature functions gives
	\begin{equation}
		\begin{gathered}
			k_1=-\frac{d_0}{2}(p-2q)(p-q-1)\delta^{p-2}
			+o(\delta^{p-2}),\\
			k_2,k_3=O(\delta^{p-1}),\quad
			k_4\to L^{-2}.
		\end{gathered}
		\label{examplek1234}
	\end{equation}
	For $p>2$, $k_1\to0$.
	For $p=2$, the leading term gives
	$k_1\to-d_0(q-1)^2$.
	This proves Eq.~\eqref{eq:power-law-beta}.
	
	Finally, $Q=k_3-k_2=-D\psi'/r$ and $a^2=(-D)\ee^{2\psi}$ give, throughout the complete range,
	\begin{equation}
		\begin{aligned}
			\frac{Q}{a^2}
			&=-\frac{q}{LA_0^2}\delta^{2q-1}[1+o(1)],\\
			a^2Q
			&=-\frac{d_0^2A_0^2q}{L}
			\delta^{2p-2q-1}[1+o(1)].
		\end{aligned}
		\label{a2Q}
	\end{equation}
	Since $q\geq1$, $Q/a^2\to0$, which gives bounded parallel-propagated curvature along timelike geodesics.
	For radial null geodesics, bounded curvature also requires $a^2Q$ to remain bounded.
	This holds exactly when $q\leq p-1/2$.
	At equality, $a^2Q$ tends to a finite nonzero limit.
	
	The additional integrability condition for nonradial null geodesics follows from monotonicity:
	\begin{equation}
		\int_{\tau_0}^{\infty}
		\frac{|\dot r|}{\sqrt{1+a^2}}\,\dd\tau
		\leq \int_{\tau_0}^{\infty}|\dot r|\,\dd\tau
		=r(\tau_0)-L<\infty.
	\end{equation}
	Hence $q\leq p-1/2$ is sufficient for bounded parallel-propagated curvature along all null geodesics.
	For null-type ends the bound holds automatically, whereas for spacelike-type ends it gives Eq.~\eqref{pqp}.
	
	\section{Exterior energy conditions}
	\label{app2}
	
	\subsection{Asymptotic conditions}
	\label{app21}
	
	Let $B=D\ee^{2\psi}$.
	Its first two derivatives are
	\begin{equation}
		\begin{gathered}
			B^\prime=\ee^{2\psi}(D^\prime+2D\psi^\prime),\\
			B^{\prime\prime}=\ee^{2\psi}
			\left(D^{\prime\prime}+4D^\prime\psi^\prime+2D\psi^{\prime\prime}+4D\psi^{\prime 2}\right).
		\end{gathered}
	\end{equation}
	If Eq.~\eqref{eq:exterior-asymptotics} holds, these identities give
	\begin{equation}
		\begin{gathered}
			B=1-\frac{2M}{r}+o(r^{-1}),\\
			B^\prime=\frac{2M}{r^2}+o(r^{-2}),\qquad
			B^{\prime\prime}=-\frac{4M}{r^3}+o(r^{-3}).
		\end{gathered}
	\end{equation}
	
	Conversely, suppose that $D$ and $B$ both equal $1-2M/r+o_2(r^{-1})$.
	Since $D\to1$,
	\begin{equation}
		\frac{B-D}{D}=o_2(r^{-1}).
	\end{equation}
	Applying the chain rule to
	\begin{equation}
		\psi=\frac12\ln\frac BD
		=\frac12\ln\left(1+\frac{B-D}{D}\right)
	\end{equation}
	then gives $\psi=o_2(r^{-1})$.
	This proves the stated equivalence.
	
	\subsection{Exact energy inequalities}
	\label{app22}
	
	For a diagonal type-I stress tensor, the dominant energy condition is
	\begin{equation}
		\rho\geq0,\qquad |p_r|\leq\rho,\qquad |p_\perp|\leq\rho.
	\end{equation}
	It implies the null and weak energy conditions.
	With $\Sigma=\rho+p_r$, the radial inequalities become
	\begin{equation}
		0\leq\Sigma\leq2\rho.
	\end{equation}
	The additional inequality required by the strong energy condition is
	\begin{equation}
		\Sigma+2p_\perp\geq0.
	\end{equation}
	Since $\Sigma\leq2\rho$, its lower bound $p_\perp\geq-\Sigma/2$ also implies $p_\perp\geq-\rho$.
	Thus all four energy conditions are equivalent to Eq.~\eqref{eq:exterior-all-EC}.
	
	\subsection{Energy conditions at large radius}
	\label{app23}
	
	For the expansion \eqref{eq:exterior-D},
	\begin{equation}
		\rho=\frac{n\mathcal P}{r^2}
		+o(\mathcal P/r^2).
	\end{equation}
	If the radial energy inequalities hold on a sufficiently distant interval, then
	\begin{equation}
		0\leq\psi'
		\leq\frac{r\rho}{D}
		=\frac{n\mathcal P}{r}[1+o(1)].
	\end{equation}
	For every $\epsilon>0$, there is a radius $R_\epsilon$ such that
	\begin{equation}
		0\leq\psi'(s)
		\leq(1+\epsilon)\frac{n\mathcal P(s)}{s}
		\qquad (s\geq R_\epsilon).
	\end{equation}
	Using $\psi(\infty)=0$ and integrating gives
	\begin{equation}
		\begin{aligned}
			0\leq-\psi(r)
			&=\int_r^\infty\psi'(s)\,\dd s\\
			&\leq(1+\epsilon)\frac{n}{n+1}\mathcal P(r),
			\qquad r\geq R_\epsilon.
		\end{aligned}
	\end{equation}
	Consequently,
	\begin{equation}
		0\leq-\psi(r)
		\leq\frac{n}{n+1}\mathcal P(r)[1+o(1)].
	\end{equation}
	This gives $\psi=O(\mathcal P)$.
	
	\begin{theorem}[]
		\label{thm6}
		Suppose $M>0$, $\alpha>0$, and $n>0$, and let $D$ and $\psi$
		satisfy Eqs.~\eqref{eq:exterior-D} and
		\eqref{eq:exterior-psi}.
		If $c_2$ satisfies Eq.~\eqref{eq:tau-range}, there is a finite
		radius $R_{\mathrm{EC}}>r_h$ such that the null, weak,
		dominant, and strong energy conditions all hold for
		every $r>R_{\mathrm{EC}}$.
	\end{theorem}
	
	\noindent\textit{Proof.}
	The correction $\mathcal P$ satisfies
	\begin{equation}
		r\mathcal P'=-(n+1)\mathcal P,\qquad
		r^2\mathcal P''=(n+1)(n+2)\mathcal P.
	\end{equation}
	The derivative bounds on the remainder in $D$ therefore give
	\begin{equation}
		\rho=\rho_0+o(\rho_0),
		\qquad
		\rho_0=\frac{n\mathcal P}{r^2}>0.
	\end{equation}
	Similarly, Eq.~\eqref{eq:exterior-psi} gives
	\begin{equation}
		\begin{gathered}
			\psi'
			=\frac{c_1(n+1)\mathcal P}{r}
			+o(\mathcal P/r),\\
			\psi''
			=-\frac{c_1(n+1)(n+2)\mathcal P}{r^2}
			+o(\mathcal P/r^2).
		\end{gathered}
	\end{equation}
	Substituting into Eq.~\eqref{eq:exterior-stress} yields
	\begin{equation}
		\begin{gathered}
			p_r
			=\left[-n+2c_1(n+1)\right]
			\frac{\mathcal P}{r^2}
			+o(\mathcal P/r^2),\\
			p_\perp
			=\left[\frac{n(n+1)}{2}-c_1(n+1)^2\right]
			\frac{\mathcal P}{r^2}
			+o(\mathcal P/r^2).
		\end{gathered}
	\end{equation}
	Using $c_1=nc_2/[2(n+1)]$, we obtain
	Eq.~\eqref{eq:exterior-pressure-limits}.
	
	The five quantities in the combined inequalities
	\eqref{eq:exterior-all-EC} have the normalized limits
	\begin{equation}
		\begin{gathered}
			\frac{\rho}{\rho_0}\longrightarrow1,\\
			\frac{\Sigma}{\rho_0}\longrightarrow c_2,\\
			\frac{2\rho-\Sigma}{\rho_0}\longrightarrow2-c_2,\\
			\frac{\rho-p_\perp}{\rho_0}
			\longrightarrow
			1-\frac{n+1}{2}(1-c_2),\\
			\frac{p_\perp+\Sigma/2}{\rho_0}
			\longrightarrow
			\frac{n+1-nc_2}{2}.
		\end{gathered}
	\end{equation}
	All five limits are strictly positive precisely when
	\begin{equation}
		c_2>0,\quad
		c_2<2,\quad
		c_2>\frac{n-1}{n+1},\quad
		c_2<\frac{n+1}{n}.
	\end{equation}
	These inequalities give Eq.~\eqref{eq:tau-range}.
	
	For $c_2$ in this open interval, all five limits are strictly positive.
	Hence there is a finite radius $R_{\mathrm{EC}}>r_h$ such that all five quantities are positive for every $r>R_{\mathrm{EC}}$.
	Equation~\eqref{eq:exterior-all-EC} then gives all four energy conditions.
	\hfill$\square$
	
	\section{Verification of the global constructions}
	\label{3}

	\subsection{Asymptotic expansions of the analytic metrics}
	\label{app31}
	
	We first verify the exterior expansion of the metric in Eq.~\eqref{eq:analytic-metric}.
	For $r>\max(r_h,2M)$, let
	\begin{equation}
		S(r)=\ln\frac{1-2M/r}{1-r_h/r}
		=\sum_{j=1}^{\infty}
		\frac{r_h^j-(2M)^j}{j r^j}.
	\end{equation}
	Only its finite sum $S_N$ enters the metric; the logarithm itself is needed only at large $r$.
	Similarly,
	\begin{equation}
		H_N(r)=\sum_{j=N+1}^{\infty}\frac{L^j}{j r^j},
		\qquad r>L.
	\end{equation}
	Both series are convergent power series in $1/r$ near infinity.
	Thus, for $j=0,1,2$,
	\begin{equation}
		\begin{aligned}
			(S_N-S)^{(j)} &= O(r^{-N-1-j}),\\
			H_N^{(j)}     &= O(r^{-N-1-j}),\\
			F_h^{(j)}     &= O(r^{-m-j}).
		\end{aligned}
	\end{equation}
	Since $N>n$ and $m>n+1$, all three functions are $o_2(r^{-n-1})$ in these estimates.
	It follows that
	\begin{equation}
		U=S+\mathcal P+o_2(r^{-n-1}),\qquad
		\psi=-c_1\mathcal P+o_2(r^{-n-1}).
	\end{equation}
	Using $(1-r_h/r)\ee^S=1-2M/r$, we obtain
	\begin{equation}
		\begin{aligned}
			D
			&=\left(1-\frac{2M}{r}\right)
			\left[1+\mathcal P+o_2(r^{-n-1})\right]\\
			&=1-\frac{2M}{r}+\mathcal P+o_2(r^{-n-1}).
		\end{aligned}
	\end{equation}
	Here $\mathcal P^2$ and $\mathcal P/r$ are both $o_2(r^{-n-1})$.
	This proves Eq.~\eqref{eq:constructed-exterior}, so the exterior result of Sec.~\ref{sec4} applies.
	
	Near the inner end, put $\delta=r-L$.
	Then
	\begin{equation}
		H_N(L+\delta)
		=\ln\frac{L}{\delta}-\sum_{j=1}^{N}\frac{1}{j}
		+h(\delta),
		\label{eq:HN-inner}
	\end{equation}
	where $h$ is analytic near $\delta=0$ and $h(0)=0$.
	The other functions $S_N$, $\mathcal P$, and $F_h$ are analytic and finite at $r=L$.
	Moreover, $1-r_h/L<0$.
	Substitution into Eq.~\eqref{eq:analytic-metric} therefore gives
	\begin{equation}
		\begin{aligned}
			D&=-d_0\delta^p[1+u(\delta)],\\
			\ee^\psi&=A_0\delta^{-q}[1+v(\delta)],
		\end{aligned}
	\end{equation}
	with $d_0,A_0>0$.
	The remainders $u,v$ are analytic near zero and vanish there.
	In particular, each satisfies
	\begin{equation}
		f\to0,\qquad \delta f'\to0,\qquad \delta^2f''\to0
		\qquad (\delta\to0^+).
	\end{equation}
	This gives the derivative control required in Sec.~\ref{sec36}.
	
	Finally, $\ee^U>0$ on $(L,\infty)$, so the only zero of $D$ is at $r_h$.
	It is simple, since
	\begin{equation}
		D'(r_h)=\frac{\ee^{U(r_h)}}{r_h}>0.
	\end{equation}
	Equation~\eqref{eq:analytic-ch} gives
	\begin{equation}
		\frac12 D'(r_h)\ee^{\psi(r_h)}
		=\frac{\ee^{U(r_h)+\psi(r_h)}}{2r_h}
		=\kappa.
	\end{equation}
	
	\subsection{Smooth matching}
	\label{app32}
	
	An explicit cutoff for Eq.~\eqref{eq:smooth-metric} is
	\begin{equation}
		\begin{gathered}
			\eta(x)=
			\begin{cases}
				0, & x\leq0,\\
				\ee^{-1/x}, & x>0,
			\end{cases}
			\qquad x=\frac{r-r_-}{r_+-r_-},\\
			\chi(r)=\frac{\eta(1-x)}{\eta(1-x)+\eta(x)}.
		\end{gathered}
	\end{equation}
	The denominator is positive for every $x$.
	The function $\chi$ is smooth, equals one for $r\leq r_-$, and vanishes for $r\geq r_+$.
	All its positive-order derivatives vanish at both matching radii.
	Hence $D$ and $\psi$ match smoothly at $r_-$ and $r_+$.
	
	Near $L$, $\chi=1$, so Eq.~\eqref{eq:HN-inner} gives the same controlled powers of $\delta$ as in the analytic construction.
	The exponential multiplying $1-2M/r$ is everywhere positive; hence $2M$ is the only horizon and is a simple zero of $D$.
	Writing $D=(1-2M/r)\ee^U$, we find
	\begin{equation}
		U(2M)+\psi(2M)
		=\chi(2M)\big[(q-p)H_N(2M)+c_s\big]=0.
	\end{equation}
	Thus
	\begin{equation}
		\frac12 D'(2M)\ee^{\psi(2M)}=\frac{1}{4M},
	\end{equation}
	independently of $r_-$ and $r_+$.
	For $r\geq r_+$, $\chi=0$, and Eq.~\eqref{eq:smooth-metric} gives the Schwarzschild metric exactly.
	
	\subsection{Extension across the horizon}
	\label{app33}
	
	For both constructions, we recover the global coordinate $\zeta$ of Sec.~\ref{sec2} from
	\begin{equation}
		\frac{\dd\ln|\zeta|}{\dd r}
		=\frac{2\kappa}{D\ee^\psi}.
		\label{eq:construction-zeta}
	\end{equation}
	Since $D$ has a simple zero at $r_h$ and $D'(r_h)\ee^{\psi(r_h)}=2\kappa$, the function
	\begin{equation}
		b(r)=\frac{2\kappa}{D(r)\ee^{\psi(r)}}-\frac{1}{r-r_h}
	\end{equation}
	extends smoothly across $r_h$.
	In the analytic construction, this extension is real analytic.
	For any $C>0$, set
	\begin{equation}
		\zeta(r)=C(r-r_h)
		\exp\!\left[\int_{r_h}^{r} b(s)\,\dd s\right].
	\end{equation}
	Then $\zeta'(r_h)=C>0$.
	Away from the horizon, Eq.~\eqref{eq:construction-zeta} gives $\zeta'=2\kappa\zeta/(D\ee^\psi)>0$, since $\zeta$ and $D$ have the same sign.
	Thus $\zeta$ is strictly increasing throughout the radial domain.
	
	The global metric coefficient
	\begin{equation}
		A=\frac{D\ee^{2\psi}}{\kappa^2\zeta}
	\end{equation}
	is positive and smooth at the horizon, because $D$ and $\zeta$ both have simple zeros there.
	At spatial infinity, $D\ee^\psi\to1$, so Eq.~\eqref{eq:construction-zeta} gives $\zeta\to\infty$.
	At the inner end, the criterion based on $I_\star$ gives $\zeta\to-\infty$ for the null-type range and $\zeta\to\zeta_i$ with $-\infty<\zeta_i<0$ for the spacelike-type range.
	Both constructions therefore have the global form assumed in Sec.~\ref{sec2}.
	
	\section{Geometry and curvature of the BV inner end}
	\label{app:BV-geometry}
	
	We consider the nonconstant interior trajectory that approaches the Nariai-type solution reported by BV \cite{boehmer2007loop}.
	Subscripts $T$ denote differentiation with respect to the BV time coordinate, and a dot denotes differentiation with respect to the future proper time $\tau$ of Eq.~\eqref{eq:BV-identification}.
	
	\subsection{Reduction of the effective equations}
	\label{app:BV-reduction}
	
	On an interior interval with $p_c>0$ and $p_b\ne0$, define
	\begin{equation}
		\begin{aligned}
			u&=b\delta_b, & \delta_b&=\frac{\sqrt{\Delta}}{\sqrt{p_c}},\\
			v&=c\delta_c, & \delta_c&=\frac{\sqrt{\Delta p_c}}{p_b}.
		\end{aligned}
		\label{eq:BV-angles}
	\end{equation}
	The identity $\delta_c p_b=\delta_b p_c$ allows us to factor the effective Hamiltonian of Ref.~\cite{boehmer2007loop} as
	\begin{equation}
		H_{\rm eff}
		=-\frac{Np_b}{2G\gamma^2\sqrt{p_c}\,\delta_b^2}
		\left(2\sin u\sin v+\sin^2u+\gamma^2\delta_b^2\right).
	\end{equation}
	The Hamiltonian constraint sets the expression in parentheses to zero.
	Since $\gamma^2\delta_b^2=\gamma^2\Delta/p_c>0$,
	neither $\sin u$ nor $\sin v$ can vanish. Writing $d=\gamma^2\delta_b^2$, we obtain
	\begin{equation}
		\frac{d}{\sin u}=-2\sin v-\sin u.
		\label{eq:BV-constraint}
	\end{equation}
	BV choose $N=\gamma\sqrt{p_c}\,\delta_b/\sin u$, which gives Eq.~\eqref{eq:BV-lapse}.
	With this choice, the Hamiltonian generating evolution in $T$ is
	\begin{equation}
		H_{\rm eff}=-\frac{p_b}{2G\gamma\delta_b}
		\left(2\sin v+\sin u+\frac{d}{\sin u}\right).
		\label{eq:BV-ham-T}
	\end{equation}
	Using $\{b,p_b\}=G\gamma$ and $\{c,p_c\}=2G\gamma$, Hamilton's equations for $p_b$ and $p_c$ give
	\begin{equation}
		\begin{aligned}
			\frac{p_{c,T}}{p_c}&=2\cos v,\\
			\frac{p_{b,T}}{p_b}
			&=\frac12\cos u\left(1-\frac{d}{\sin^2u}\right)\\
			&=\cos u\left(1+\frac{\sin v}{\sin u}\right)
			\equiv F(u,v).
		\end{aligned}
		\label{eq:BV-triad-evolution}
	\end{equation}
	The last equality uses Eq.~\eqref{eq:BV-constraint}.
	It follows from Eq.~\eqref{eq:BV-angles} that
	\begin{equation}
		\frac{\delta_{b,T}}{\delta_b}=-\cos v,\quad
		\frac{\delta_{c,T}}{\delta_c}=\cos v-F.
		\label{eq:BV-delta-evolution}
	\end{equation}
	
	The connection equations
	$b_T=G\gamma\,\partial_{p_b}H_{\rm eff}$ and
	$c_T=2G\gamma\,\partial_{p_c}H_{\rm eff}$ can be written as
	\begin{equation}
		\begin{gathered}
			\delta_b b_T
			=-\tfrac12(\sin u+d/\sin u)-\sin v+v\cos v,\\
			\delta_c c_T
			=-\sin v-\tfrac12\sin u-v\cos v
			+uF+\tfrac12 d/\sin u.
		\end{gathered}
	\end{equation}
	Substitution of Eq.~\eqref{eq:BV-constraint} gives
	\begin{equation}
		\begin{aligned}
			\delta_b b_T&=v\cos v,\\
			\delta_c c_T&=-2\sin v-\sin u-v\cos v+uF.
		\end{aligned}
	\end{equation}
	We can now differentiate $u=b\delta_b$ and $v=c\delta_c$.
	Using Eq.~\eqref{eq:BV-delta-evolution}, we find
	\begin{equation}
		\begin{aligned}
			u_T
			&=\delta_b b_T+u\frac{\delta_{b,T}}{\delta_b}
			=(v-u)\cos v\equiv f(u,v),\\
			v_T
			&=\delta_c c_T+v\frac{\delta_{c,T}}{\delta_c}\\
			&=-2\sin v-\sin u+(u-v)F(u,v)\equiv g(u,v).
		\end{aligned}
		\label{eq:BV-reduced}
	\end{equation}
	The right-hand sides depend only on $u$ and $v$, so these two equations can be solved independently of $p_b$.
	To reconstruct the original variables, the constraint fixes $p_c$, while $p_b$ requires one integration constant:
	\begin{equation}
		\begin{gathered}
			p_c=\frac{\gamma^2\Delta}
			{-2\sin u\sin v-\sin^2u},\\
			p_b(T)=p_b(T_0)
			\exp\left[\int_{T_0}^{T}F(u(s),v(s))\,\dd s\right].
		\end{gathered}
	\end{equation}
	Then $b=u/\delta_b$ and $c=v/\delta_c$.
	In particular, the metric functions satisfy the exact relations
	\begin{equation}
		\frac{r_T}{r}=\cos v,\quad
		\frac{a_T}{a}=F-\cos v,\quad
		N=\frac{\gamma\sqrt{\Delta}}{\sin u}.
		\label{eq:BV-metric-rates}
	\end{equation}
	
	\subsection{The fixed point and the approach of the trajectory}
	\label{app:BV-focus}
	
	We first seek a solution with constant areal radius $r(T)=L>0$.
	Since $r_T=r\cos v$, we must have $\cos v(T)=0$ for all $T$.
	The zeros of the cosine are isolated, so continuity requires $v(T)$ to be constant.
	The equation $u_T=(v-u)\cos v$ then
	implies that $u(T)$ is constant as well.
	Thus a constant-radius solution corresponds to a fixed point of the reduced equations.
	We denote its coordinates by $(u_*,v_*)$.
	Since $r=L>0$ is constant, the equation $r_T=r\cos v$ gives $\cos v_*=0$.
	For the solution studied by BV, Eq.~(64) of Ref.~\cite{boehmer2007loop} gives
	\begin{equation}
		v_*=\sqrt{\Delta\bar p_c}\,\frac{\bar c}{\bar p_b}
		=-\frac{\pi}{2}.
	\end{equation}
	At $v=v_*$, the first equation in Eq.~\eqref{eq:BV-reduced}
	gives $u_T=0$. For $v$ to remain constant, we must also
	have $g(u_*,v_*)=0$. Since $F(u,v_*)=\cos u-\cot u$,
	this condition reads
	\begin{equation}
		2-\sin u_*+
		(u_*+\pi/2)(\cos u_*-\cot u_*)=0.
		\label{eq:BV-fixed-root}
	\end{equation}
	The root corresponding to the BV solution lies between $0.58$ and $0.59$ and is $u_*\simeq0.58723252$.
	At this fixed point, $u$, $v$, and the sphere radius $r=L$ are constant, while the scale factor in the $x$ direction grows as $a(T)\propto\ee^{-\alpha T}$ toward the inner end $T\to-\infty$.
	Indeed,
	\begin{equation}
		\frac{a_T}{a}=\frac{p_{b,T}}{p_b}
		=F(u_*,v_*)=-\alpha.
	\end{equation}
	
	Let $J$ be the matrix of first derivatives of $(f,g)$ with respect to $(u,v)$ at the fixed point.
	Write $F_*=F(u_*,v_*)$.
	Differentiating $f=(v-u)\cos v$ gives
	\begin{equation}
		J_{11}=-\cos v_*=0,\quad
		J_{12}=-(u_*+\pi/2).
	\end{equation}
	For the second row, $(\partial_vF)_*=0$ and
	$(\partial_uF)_*=-\sin u_*+1/\sin^2u_*$.
	Thus
	\begin{equation}
		\begin{gathered}
			J_{21}=-\cos u_*+F_*+
			(u_*+\pi/2)(\partial_uF)_*,\\
			J_{22}=-F_*.
		\end{gathered}
	\end{equation}
	The full matrix is therefore
	\begin{equation}
		J=
		\left.
		\begin{pmatrix}
			\partial_u f&\partial_v f\\
			\partial_u g&\partial_v g
		\end{pmatrix}
		\right|_{(u_*,v_*)}
		\simeq
		\begin{pmatrix}
			0&-2.15803\\
			4.33165&0.67003
		\end{pmatrix}.
	\end{equation}
	Its trace is exactly $\alpha$.
	The eigenvalues are $\nu\pm i\omega$, where
	\begin{equation}
		\nu=\frac{\alpha}{2}>0,\quad
		\omega=\sqrt{\det J-\nu^2}\simeq3.03901.
	\end{equation}
	
	We now determine the rate at which the nonlinear trajectory approaches the fixed point.
	Translate $(u_*,v_*)$ to the origin.
	Since the eigenvalues form a nonreal conjugate pair, a real invertible linear change of coordinates brings $J$ to the form below.
	Let $Y=(u-u_*,v-v_*)^{\mathsf T}$.
	Since $f=g=0$ at the fixed point, Taylor expansion of the reduced equations gives
	\begin{equation}
		Y_T=JY+O(|Y|^2).
	\end{equation}
	The eigenvalues of $J$ are $\nu\pm i\omega$. We therefore choose real coordinates $Y=Pz$ in which
	\begin{equation}
		P^{-1}JP=
		\begin{pmatrix}\nu&-\omega\\ \omega&\nu\end{pmatrix}.
	\end{equation}
	Writing $z=(z_1,z_2)^{\mathsf T}$ and $s=|z|$, we obtain
	\begin{equation}
		\frac{\dd}{\dd T}
		\begin{pmatrix}z_1\\z_2\end{pmatrix}
		=
		\begin{pmatrix}
			\nu&-\omega\\
			\omega&\nu
		\end{pmatrix}
		\begin{pmatrix}z_1\\z_2\end{pmatrix}
		+\begin{pmatrix}R_1\\R_2\end{pmatrix},
	\end{equation}
	where $R_1,R_2=O(s^2)$.
	
	Introduce polar coordinates $z_1=s\cos\vartheta$ and $z_2=s\sin\vartheta$.
	Then
	\begin{equation}
		\begin{gathered}
			(\ln s)_T
			=\frac{z_1z_{1,T}+z_2z_{2,T}}{s^2}
			=\nu+O(s),\\
			\vartheta_T
			=\frac{z_1z_{2,T}-z_2z_{1,T}}{s^2}
			=\omega+O(s).
		\end{gathered}
		\label{eq:BV-spiral}
	\end{equation}
	The terms containing $R_1$ and $R_2$ are $O(s^3)$ in each numerator and hence contribute $O(s)$ after division by $s^2$.
	
	Since $s\to0$ as $T\to-\infty$, we can choose $T_1$ such that $(\ln s)_T\geq\nu/2$ for $T\leq T_1$.
	Integration from $T$ to $T_1$ gives
	\begin{equation}
		s(T)\leq s(T_1)
		\exp\left[\frac{\nu}{2}(T-T_1)\right].
	\end{equation}
	Hence $\int_{-\infty}^{T_1}s(T)\,\dd T<\infty$.
	Since the $O(s)$ terms in Eq.~\eqref{eq:BV-spiral} are integrable, integration gives, as $T\to-\infty$,
	\begin{equation}
		s(T)\sim C\ee^{\nu T},\qquad
		\vartheta(T)=\omega T+\vartheta_0+o(1),
		\label{eq:BV-nonlinear-rate}
	\end{equation}
	where $C>0$ and $\vartheta_0$ are constants.
	Thus Eq.~\eqref{eq:BV-nonlinear-rate} gives the rate along the nonconstant trajectory itself.
	Similarly, $F-\cos v=-\alpha+O(s)$ near the fixed point.
	Integrating $a_T/a=F-\cos v$ gives $a\sim a_0\ee^{-\alpha T}$ for some $a_0>0$, in agreement with Eq.~\eqref{eq:BV-asymptotics}.
	
	The line $v=v_*$ remains a line through the origin in the $(z_1,z_2)$ plane.
	For sufficiently negative $T$, $\vartheta_T>0$.
	Thus, as $T$ decreases toward the inner end, $\vartheta$ decreases monotonically to $-\infty$.
	The trajectory crosses the line $v=v_*$ once every half turn, giving a sequence of crossing times $T_n\to-\infty$.
	
	\subsection{Scalar and parallel-propagated curvature}
	\label{app:BV-curvature}
	
	We first establish the scalar curvature limits on the oscillatory trajectory.
	Eq.~\eqref{eq:BV-metric-rates} and
	$\dd\tau=-N\,\dd T$ give the exact relations
	\begin{equation}
		H_r\equiv\frac{\dot r}{r}=-\frac{\cos v}{N},
		\quad
		H_a\equiv\frac{\dot a}{a}=-\frac{F-\cos v}{N}.
	\end{equation}
	As the trajectory approaches the fixed point, $\cos v\to0$, $F\to-\alpha$, and $N\to\bar N>0$.
	Hence
	\begin{equation}
		H_r\to0,\qquad H_a\to H=\alpha/\bar N.
	\end{equation}
	Since $H_r$ and $H_a$ are smooth at the fixed point and $\dot u,\dot v\to0$, we have $\dot H_r,\dot H_a\to0$.
	The curvature functions on the oscillatory trajectory therefore satisfy
	\begin{equation}
		\begin{aligned}
			k_1&=-(\dot H_a+H_a^2)\to-H^2,\\
			k_2&=-H_aH_r\to0,\\
			k_3&=-(\dot H_r+H_r^2)\to0,\\
			k_4&=r^{-2}+H_r^2\to L^{-2}.
		\end{aligned}
	\end{equation}
	Substituting into the scalar invariants of Sec.~\ref{sec3} gives
	\begin{equation}
		\begin{aligned}
			R&\to2(H^2+L^{-2}),\\
			\mathcal S&\to2(H^4+L^{-4}),\\
			\mathcal K&\to4(H^4+L^{-4}).
		\end{aligned}
	\end{equation}
	In particular, $Q=k_3-k_2\to0$ and $Q/a^2\to0$.
	The timelike parallel-propagated curvature criterion of Sec.~\ref{sec3} therefore gives bounded curvature along each fixed timelike geodesic near the inner end.
	
	For radial null geodesics, the criterion also involves $a^2Q$.
	To prove that this quantity is unbounded, it suffices to evaluate it along the sequence of turning points $T_n$ found above.
	Differentiating $r_T=r\cos v$ gives
	\begin{equation}
		\frac{r_{TT}}{r}=\cos^2v-\sin v\,v_T.
	\end{equation}
	At a crossing of $v_*=-\pi/2$, we have $r_T=0$ and
	$r_{TT}/r=v_T$.
	Since
	\begin{equation}
		\ddot r=\frac{r_{TT}}{N^2}-\frac{r_TN_T}{N^3},
	\end{equation}
	the term containing $N_T$ vanishes there.
	Hence at each turning point,
	\begin{equation}
		Q=\frac{H_a\dot r-\ddot r}{r}
		=-\frac{v_T}{N^2}.
		\label{eq:BV-Q-turn}
	\end{equation}
	
	Along the line $v=v_*$, Taylor expansion of the second reduced equation gives
	\begin{equation}
		v_T=g(u,v_*)
		=J_{21}(u-u_*)+O\big((u-u_*)^2\big),
	\end{equation}
	with $J_{21}\simeq4.33165\ne0$.
	On this line, $s$ is a fixed positive multiple of $|u-u_*|$.
	It follows from Eq.~\eqref{eq:BV-Q-turn} and $N\to\bar N>0$ that $|Q(T_n)|/s(T_n)$ tends to a finite, strictly positive constant.
	Equation~\eqref{eq:BV-nonlinear-rate}, with
	$\nu=\alpha/2$, now gives
	\begin{equation}
		|Q(T_n)|\sim C_Q\ee^{\alpha T_n/2},
		\qquad C_Q>0.
	\end{equation}
	Successive crossings lie on opposite sides of $u_*$ along the line $v=v_*$.
	Thus $v_T$ and $Q$ alternate in sign for all sufficiently large $n$.
	By Eq.~\eqref{eq:BV-Q-turn}, the corresponding turning points are alternating nondegenerate minima and maxima of $r$.
	
	Although $Q(T_n)$ tends to zero, it decays more slowly than $a(T_n)^{-2}$.
	Indeed, $a(T_n)\sim a_0\ee^{-\alpha T_n}$ yields
	\begin{equation}
		|a^2Q|(T_n)
		\sim a_0^2C_Q\exp\left(-\frac{3\alpha T_n}{2}\right)
		\longrightarrow\infty.
	\end{equation}
	This proves Eq.~\eqref{eq:BV-a2Q-unbounded}.
	
	By the radial null curvature criterion of Sec.~\ref{sec3}, this sequence is sufficient to establish that a parallel-propagated curvature component is unbounded along each radial null geodesic approaching the inner end.
	
	\bibliography{ATNC}

@misc{LiangLi2026AT,
	author        = {Liang, Yi-Bo and Li, Hong-Rong},
	title         = {Asymptotic Throat: The Geometric Inevitability
	of Regular Black Holes},
	year          = {2026},
	eprint        = {2604.11396},
	archivePrefix = {arXiv},
	primaryClass  = {gr-qc},
	doi           = {10.48550/arXiv.2604.11396}
}

@article{Schwarzschild1916,
	author  = {Schwarzschild, Karl},
	title   = {{\"U}ber das Gravitationsfeld eines Massenpunktes
	nach der Einsteinschen Theorie},
	journal = {Sitzungsber. Preuss. Akad. Wiss. Berlin
	(Math. Phys.)},
	pages   = {189--196},
	year    = {1916}
}

@article{penrose1965gravitational,
	title={Gravitational collapse and space-time singularities},
	author={Penrose, Roger},
	journal={Physical Review Letters},
	volume={14},
	number={3},
	pages={57},
	year={1965},
	publisher={APS}
}

@article{hawking1967occurrence,
	title={The occurrence of singularities in cosmology. III. Causality and singularities},
	author={Hawking, Stephen William},
	journal={Proceedings of the Royal Society of London. Series A. Mathematical and Physical Sciences},
	volume={300},
	number={1461},
	pages={187--201},
	year={1967},
	publisher={The Royal Society London}
}

@article{hawking1970singularities,
	title={The singularities of gravitational collapse and cosmology},
	author={Hawking, Stephen William and Penrose, Roger},
	journal={Proceedings of the Royal Society of London. A. Mathematical and Physical Sciences},
	volume={314},
	number={1519},
	pages={529--548},
	year={1970},
	publisher={The Royal Society London}
}

@inproceedings{bardeen1968non,
	title={Non-singular general relativistic gravitational collapse},
	author={Bardeen, James},
	booktitle={Proceedings of the 5th International Conference on Gravitation and the Theory of Relativity},
	pages={87},
	year={1968}
}

@article{ayon1998regular,
	title={Regular black hole in general relativity coupled to nonlinear electrodynamics},
	author={Ayon-Beato, Eloy and Garcia, Alberto},
	journal={Physical review letters},
	volume={80},
	number={23},
	pages={5056},
	year={1998},
	publisher={APS}
}

@article{ayon2000bardeen,
	title={The Bardeen model as a nonlinear magnetic monopole},
	author={Ayon-Beato, Eloy and Garc{\i}a, Alberto},
	journal={Physics Letters B},
	volume={493},
	number={1-2},
	pages={149--152},
	year={2000},
	publisher={Elsevier}
}

@article{bronnikov2001regular,
	title={Regular magnetic black holes and monopoles from nonlinear electrodynamics},
	author={Bronnikov, Kirill A},
	journal={Physical Review D},
	volume={63},
	number={4},
	pages={044005},
	year={2001},
	publisher={APS}
}

@article{burinskii2002new,
	title={New type of regular black holes and particlelike solutions from nonlinear electrodynamics},
	author={Burinskii, Alexander and Hildebrandt, Sergi R},
	journal={Physical Review D},
	volume={65},
	number={10},
	pages={104017},
	year={2002},
	publisher={APS}
}

@article{fan2016construction,
	title={Construction of regular black holes in general relativity},
	author={Fan, Zhong-Ying and Wang, Xiaobao},
	journal={Physical Review D},
	volume={94},
	number={12},
	pages={124027},
	year={2016},
	publisher={APS}
}

@incollection{bronnikov2023regular,
	title={Regular black holes sourced by nonlinear electrodynamics},
	author={Bronnikov, Kirill A},
	booktitle={Regular Black Holes: Towards a New Paradigm of Gravitational Collapse},
	pages={37--67},
	year={2023},
	publisher={Springer}
}

@article{simpson2019black,
	title={Black-bounce to traversable wormhole},
	author={Simpson, Alex and Visser, Matt},
	journal={Journal of Cosmology and Astroparticle Physics},
	volume={2019},
	number={02},
	pages={042},
	year={2019},
	publisher={IOP Publishing}
}

@article{lobo2021novel,
	title={Novel black-bounce spacetimes: wormholes, regularity, energy conditions, and causal structure},
	author={Lobo, Francisco SN and Rodrigues, Manuel E and Silva, Marcos V de S and Simpson, Alex and Visser, Matt},
	journal={Physical Review D},
	volume={103},
	number={8},
	pages={084052},
	year={2021},
	publisher={APS}
}

@article{bronnikov2022black,
	title={Black bounces, wormholes, and partly phantom scalar fields},
	author={Bronnikov, KA},
	journal={Physical Review D},
	volume={106},
	number={6},
	pages={064029},
	year={2022},
	publisher={APS}
}

@article{bronnikov2022field,
	title={Field sources for Simpson-Visser spacetimes},
	author={Bronnikov, Kirill A and Walia, Rahul Kumar},
	journal={Physical Review D},
	volume={105},
	number={4},
	pages={044039},
	year={2022},
	publisher={APS}
}

@article{canate2022black,
	title={Black bounces as magnetically charged phantom regular black holes in Einstein-nonlinear electrodynamics gravity coupled to a self-interacting scalar field},
	author={Ca{\~n}ate, Pedro},
	journal={Physical Review D},
	volume={106},
	number={2},
	pages={024031},
	year={2022},
	publisher={APS}
}

@article{Zaslavskii2007,
	author        = {Zaslavskii, O. B.},
	title         = {Truly naked spherically symmetric and distorted black holes},
	journal       = {Phys. Rev. D},
	volume        = {76},
	pages         = {024015},
	year          = {2007},
	doi           = {10.1103/PhysRevD.76.024015},
	eprint        = {0706.2727},
	archivePrefix = {arXiv},
	primaryClass  = {gr-qc}
}

@article{Bronnikov2008,
	author        = {Bronnikov, K. A. and Elizalde, E.
	and Odintsov, S. D. and Zaslavskii, O. B.},
	title         = {Horizons versus singularities in spherically symmetric space-times},
	journal       = {Phys. Rev. D},
	volume        = {78},
	pages         = {064049},
	year          = {2008},
	doi           = {10.1103/PhysRevD.78.064049},
	eprint        = {0805.1095},
	archivePrefix = {arXiv},
	primaryClass  = {gr-qc}
}

@book{Chandrasekhar1992,
	author    = {Chandrasekhar, S.},
	title     = {The Mathematical Theory of Black Holes},
	publisher = {Oxford University Press},
	address   = {Oxford},
	year      = {1992}
}

@article{Nariai1999,
	author  = {Nariai, Hidekazu},
	title   = {On a New Cosmological Solution of {Einstein}'s
	Field Equations of Gravitation},
	journal = {Gen. Relativ. Gravit.},
	volume  = {31},
	pages   = {963--971},
	year    = {1999},
	doi     = {10.1023/A:1026602724948},
	note    = {Reprint of the 1951 paper}
}

@incollection{MartinMorunoVisser2017,
	author        = {Mart{\'i}n-Moruno, Prado and Visser, Matt},
	title         = {Classical and Semi-classical Energy Conditions},
	booktitle     = {Wormholes, Warp Drives and Energy Conditions},
	editor        = {Lobo, Francisco S. N.},
	series        = {Fundamental Theories of Physics},
	volume        = {189},
	publisher     = {Springer},
	address       = {Cham},
	pages         = {193--213},
	year          = {2017},
	doi           = {10.1007/978-3-319-55182-1_9},
	eprint        = {1702.05915},
	archivePrefix = {arXiv},
	primaryClass  = {gr-qc}
}

@article{simpson2019vaidya,
	title={Vaidya spacetimes, black-bounces, and traversable wormholes},
	author={Simpson, Alex and Martin-Moruno, Prado and Visser, Matt},
	journal={Classical and Quantum Gravity},
	volume={36},
	number={14},
	pages={145007},
	year={2019},
	publisher={IOP Publishing}
}

@article{lobo2020dynamic,
	title={Dynamic thin-shell black-bounce traversable wormholes},
	author={Lobo, Francisco SN and Simpson, Alex and Visser, Matt},
	journal={Physical Review D},
	volume={101},
	number={12},
	pages={124035},
	year={2020},
	publisher={APS}
}

@article{shinkai2002fate,
	title={Fate of the first traversible wormhole: Black-hole collapse or inflationary expansion},
	author={Shinkai, Hisa-aki and Hayward, Sean A},
	journal={Physical Review D},
	volume={66},
	number={4},
	pages={044005},
	year={2002},
	publisher={APS}
}

@article{hayward2009wormhole,
	title={Wormhole dynamics in spherical symmetry},
	author={Hayward, Sean A},
	journal={Physical Review D—Particles, Fields, Gravitation, and Cosmology},
	volume={79},
	number={12},
	pages={124001},
	year={2009},
	publisher={APS}
}

@article{hayward1999dynamic,
	title={Dynamic wormholes},
	author={Hayward, Sean A},
	journal={International Journal of Modern Physics D},
	volume={8},
	number={03},
	pages={373--382},
	year={1999},
	publisher={World Scientific}
}

@article{hayward2006formation,
	title={Formation and evaporation of nonsingular black holes},
	author={Hayward, Sean A},
	journal={Physical Review Letters},
	volume={96},
	number={3},
	pages={031103},
	year={2006},
	publisher={APS}
}

@article{yang2021trapping,
	title={Trapping horizons of the evolving charged wormhole and black bounce},
	author={Yang, Jinbo and Huang, Hyat},
	journal={Physical Review D},
	volume={104},
	number={8},
	pages={084005},
	year={2021},
	publisher={APS}
}

@article{bronnikov2006regular,
	title={Regular phantom black holes},
	author={Bronnikov, Kirill A and Fabris, Julio C},
	journal={Physical Review Letters},
	volume={96},
	number={25},
	pages={251101},
	year={2006},
	publisher={APS}
}

@article{carballo2020geodesically,
	title={Geodesically complete black holes},
	author={Carballo-Rubio, Ra{\'u}l and Di Filippo, Francesco and Liberati, Stefano and Visser, Matt},
	journal={Physical Review D},
	volume={101},
	number={8},
	pages={084047},
	year={2020},
	publisher={APS}
}

@article{boehmer2007loop,
	title={Loop quantum dynamics of the Schwarzschild interior},
	author={Boehmer, Christian G and Vandersloot, Kevin},
	journal={Physical Review D},
	volume={76},
	number={10},
	pages={104030},
	year={2007},
	publisher={APS}
}

@article{liang2025globally,
	title={The Globally Trapped Future: A Fate for Black Holes and Wormholes},
	author={Liang, Yi-Bo and Li, Hong-Rong},
	journal={Classical and Quantum Gravity},
	year={2025}
}

@incollection{ashtekar2023regular,
	title={Regular black holes from loop quantum gravity},
	author={Ashtekar, Abhay and Olmedo, Javier and Singh, Parampreet},
	booktitle={Regular black holes: towards a new paradigm of gravitational collapse},
	pages={235--282},
	year={2023},
	publisher={Springer}
}

\end{document}